%% file: paper.tex
\documentclass[acmsmall,screen,nonacm]{acmart}
\makeatletter
\def\@ACM@checkaffil{
    \if@ACM@instpresent\else
    \ClassWarningNoLine{\@classname}{No institution present for an affiliation}%
    \fi
    \if@ACM@citypresent\else
    \ClassWarningNoLine{\@classname}{No city present for an affiliation}%
    \fi
    \if@ACM@countrypresent\else
        \ClassWarningNoLine{\@classname}{No country present for an affiliation}%
    \fi
}
\makeatother

\newcommand{\hide}[1]{}
\hide{\bibliographystyle{ACM-Reference-Format}}
\usepackage{prooftree}

\usepackage{wrapfig}

\usepackage{todonotes}
\usepackage{listings}
\definecolor{mygreen}{rgb}{0,0.4,0}
\definecolor{myblue}{rgb}{0,0,0.7}
\definecolor{myred}{rgb}{0.7,0,0}
\definecolor{mygray}{rgb}{0.5,0.5,0.5}
\definecolor{mymauve}{rgb}{0.58,0,0.82}

\newcommand{\cb}{\color{myblue}}
\usepackage{soul}
\usepackage{wrapfig}

\newcommand*{\mlstinline}[1]{\mbox{\lstinline[mathescape]|#1|}}
\newcommand*{\mlstinlin}[1]{\mbox{\lstinline|#1|}}
\newcommand{\RR}{\mathbb{R}}
\newcommand{\NN}{\mathbb{N}}
\newcommand{\dd}{\mathrm{d}}
\newcommand{\ProbKer}{\mathbf{ProbKer}}
\newcommand{\MeasKer}{\mathbf{MeasKer}}
\newcommand{\Qbs}{\mathbf{Qbs}}
\AtBeginDocument{%
  \providecommand\BibTeX{{%
    \normalfont B\kern-0.5em{\scshape i\kern-0.25em b}\kern-0.8em\TeX}}}

\setcopyright{rightsretained}
\acmPrice{}
\acmDOI{10.1145/3571239}
\acmYear{2023}
\copyrightyear{2023}
\acmSubmissionID{popl23main-p324-p}
\acmJournal{PACMPL}
\acmVolume{7}
\acmNumber{POPL}
\acmMonth{1}
\begin{document}

\title{Affine~Monads~and~Lazy~Structures for~Bayesian~Programming}

\author{Swaraj Dash}
\author{Younesse Kaddar}
\author{Hugo Paquet}
\author{Sam Staton}
\affiliation{ Univ.~Oxford%
\hide{\country{UK}}
}

\renewcommand{\shortauthors}{Swaraj Dash, Younesse Kaddar, Hugo Paquet, and Sam Staton}

\begin{abstract}
  We show that streams and lazy data structures are a natural idiom for programming with infinite-dimensional Bayesian methods such as Poisson processes, Gaussian processes, jump processes, Dirichlet processes, and Beta processes. The crucial semantic idea, inspired by developments in synthetic probability theory, is to work with two separate monads: an affine monad of probability, which supports laziness, and a commutative, non-affine monad of measures, which does not. (Affine means that $T(1)\cong 1$.) We show that the separation is important from a decidability perspective, and that the recent model of quasi-Borel spaces supports these two monads.

  To perform Bayesian inference with these examples, we introduce new inference methods that are specially adapted to laziness; they are proven correct by reference to the Metropolis-Hastings-Green method. Our theoretical development is implemented as a Haskell library, LazyPPL. 
\end{abstract}

\hide{


}
\maketitle

\section{Introduction}
\input{intro}

\section{A monadic metalanguage for probability and measure}
\input{interface}

\section{Examples taking advantage of laziness}

\input{regression}



\section{Interpretation of the probability monad using infinite trees}
\input{monads}

\section{A new Metropolis-Hastings kernel for laziness}
\input{msmh}

\section{Mixed kernels and single-site Metropolis-Hastings}
\input{singlesite}

\section{Related work on laziness and practical synthetic probability}
\input{related}

\section{Summary}
We have presented a metalanguage for lazy probabilistic programming with
two monads (for probability and measure, \S\ref{sec:interface}) and new Metropolis-Hastings-based algorithms (\S\ref{sec:msmh}, \S\ref{sec:single-site}). The methods are based on recent foundations from quasi-Borel spaces and synthetic probability theory (\S\ref{sec:monads}).

The separation into two monads is essential for decidability reasons (Theorem~\ref{thm:iid}, Prop.~\ref{prop:undecidable}), but also yields a useful programming idiom for a variety of infinite-dimensional Bayesian models, including piecewise linear regression (\S\ref{sec:easyexamples:regpiecewiselin}), non-parametric clustering (\S\ref{sec:dp}), non-parametric feature extraction (\S\ref{sec:ibp}), and Gaussian process regression (\S\ref{sec:gp}), and compositions of these. As we have shown (for instance by considering rescaling), laziness allows for compositional programming, avoiding the problem of passing around truncation bounds.

\bibliography{./bibliography}


\end{document}

%% file: intro.tex
\label{sec:intro}

Bayesian statistical models often naturally involve infinite-dimensional spaces, and in this paper we show that these can be dealt with programmatically using lazy structures.
To show this, we provide a monadic metalanguage for probabilistic programming that admits streams and other lazy data structures (\S\ref{sec:interface}).
The general key point is that an `affine' monad can support lazy programming, but a non-affine one cannot (see \S\ref{sec:intro:affine} and Theorems~\ref{thm:iid} and~\ref{prop:undecidable}). For probabilistic programming, we thus consider two monads: an affine monad of probability, and a non-affine monad of measures. 
We demonstrate the expressive compositional power of this through a wealth of examples (\S\ref{sec:intro:poisson}, \S\ref{sec:easyexamples}, \S\ref{sec:randfun}; \cite{nonanongit}). We show that these are feasible by giving new Metropolis-Hastings inference algorithms (\S\ref{sec:intro:mhg}, \S\ref{sec:msmh}, \S\ref{sec:single-site}) that work in a lazy setting. Our development is motivated by new compositional methods in categorical measure theory, such as quasi-Borel spaces: in Section~\ref{sec:monads} we give a new formulation of this together with an implementation as a Haskell library, LazyPPL~\cite{nonanongit}.

\subsection{Monte Carlo methods, Bayesian models, and unnormalized measures}
\label{sec:intro:mc}
It is often said that Monte Carlo methods are the reason for the explosion in practical Bayesian statistics over the past 30 years~(e.g.~\cite[\S1.1]{geyer}, \cite[\S1.4]{gv-bnp}). One account of Monte Carlo methods is that they are methods for sampling from a probability distribution that is specified as an \emph{unnormalized} measure, that is, a measure that is only specified up to an unknown normalizing constant (e.g.~\cite{geyer,tierney-mcposterior}). This matches the three primitive aspects of Bayesian statistics, which are:
\begin{itemize}
\item prior --- a probability measure;
\item likelihood --- often expressed by a density, or weight, contributing to the unnormalized aspect of the measure;
\item posterior --- a probability measure that is proportional to the product of the likelihood and the prior, which is what the Monte Carlo method provides samples from.
\end{itemize}
Our aim here is to explore the role of laziness in building and composing these measures. Our motivation comes from two directions: practical and theoretical.

On the practical side, probabilistic programming languages for Bayesian modelling (such as Bugs~\cite{bugs}, Church~\cite{goodman:church}, Stan~\cite{stan} and others) can often be regarded as programming languages describing unnormalized measures, that are endowed with efficient Monte Carlo samplers. Many focus on finite dimensional models, but some allow unbounded  dimension, notably Church, which is a key starting point for our work. (See~\S\ref{sec:relatedwork} for a fuller discussion of prior work.)

On the theoretical side, researchers have recently proposed categorical or synthetic accounts of probability theory~\cite{cho-jacobs,fritz} and measure theory~\cite{kock,scibior,qbs}, with the aim of developing compositional structures based on commutative and affine monads and monoidal categories. There are various aims in that work, some axiomatic, and some seeking to sidestep cumbersome issues with measure theory, such as the absence of function spaces and of a strong monad of measures (see \S\ref{sec:monads}, where we treat this).
In this way, probabilistic programming can be viewed as a \emph{practical measure theory}, with compositionality built in, and where types are spaces, and programs are suitably good measures on the spaces.
By exploring fully expressive probabilistic programming languages, we are exploring the abstract and higher-order spaces of synthetic measure theory.

\subsection{Practical illustration: the Poisson process}
\label{sec:intro:poisson}
To illustrate further on the practical side, we briefly consider a `non-parametric' model now: the one-dimensional homogenous Poisson point process. This is a random countable collection of points on the positive real line, such that within any finite interval $[a,b]$ the expected number of points is proportional to $(b-a)$, and the number of points in disjoint regions is independent.
Some draws from a Poisson point process are shown in Figure~\ref{fig:affine} (a). A Poisson point process is easy to define using laziness, and we flesh out this definition in Section~\ref{sec:poisson}.

Of course, the pictures in Figure~\ref{fig:affine} (a) each show a finite number of points, but this is because we have constrained the viewport to a finite window.
In practice we may want to use the point process as part of a larger model, and in Section~\ref{sec:easyexamples:reg} we use it as part of a regression problem. Then it is unclear where to truncate it to an arbitrary viewport in advance, and, as we demonstrate, this can break the compositionality. 
This is often the case in statistical models, as in other areas of programming: if we just focus on running whole programs, we lose perspective of the conceptual and practical building blocks. We illustrate this further with other examples of non-parametric processes including Dirichlet process clustering (\S\ref{sec:dp}) and Gaussian process regression (\S\ref{sec:gp}).

\subsection{Theoretical aspects: affine monoidal structure and synthetic spaces}
\label{sec:intro:affine}
To illustrate the theoretical side, we recall that in the categorical foundations of measure theory, a morphism $X\to I$ into the monoidal unit describes a parameterized measure on the one-point space. If we focus on normalized probability measures, there should be exactly one measure on the one point space. So in the normalized setting,~$I$ should be a terminal object, in other words, we are working with an affine monoidal category (e.g.~\cite{cho-jacobs,fritz,jacobs-probabilities,coecke-terminality,shiebler,fgp-definetti}). This is shown diagrammatically in Figure~\ref{fig:affine} (b), and
\begin{figure}
  \framebox{\includegraphics[scale=0.64]{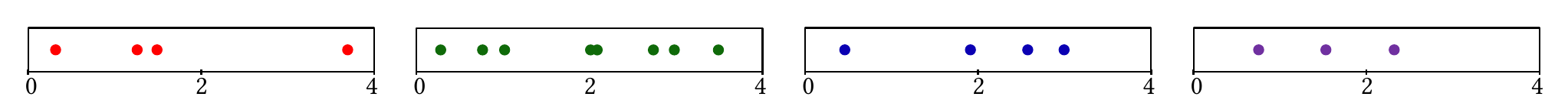}}\\[5mm]
  \framebox{\includegraphics[scale=0.5]{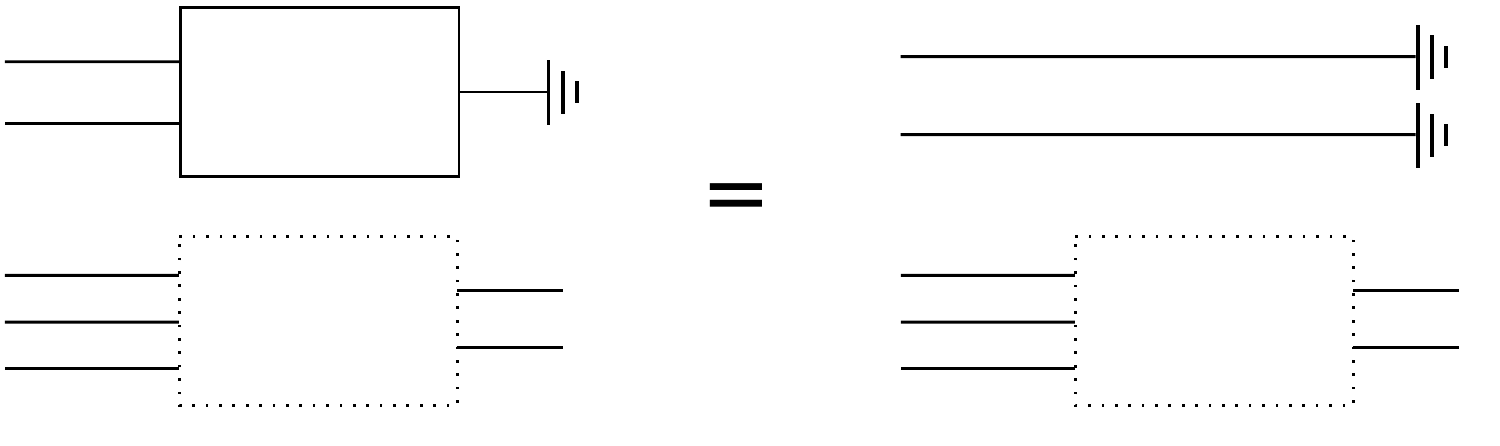}}\qquad\qquad
  \framebox{\includegraphics[scale=0.5]{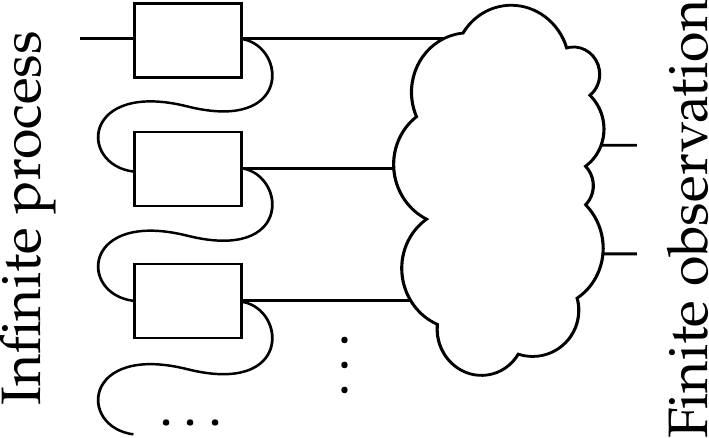}}
  \caption{\label{fig:affine}
    (a)~Four samples from a 1D Poisson point process with rate $1$, with viewport restricted to $[0,4]$.
    (b)~The law for affine monoidal categories in string diagram form. (c) Visualizing dataflow in a lazy infinite process.}
  \end{figure}
  we can regard the diagram as a dataflow diagram.
To see where laziness comes in, we regard the Poisson point process again, now as a dataflow diagram (Fig.~\ref{fig:affine} (c)).
The infinite process is on the left, and the cloud represents the plotting routine, or whatever happens next; intuitively, those morphisms on the left that are not used in what follows need never be inspected, and in that case the process can be truncated via Fig.~\ref{fig:affine} (b). Thus, affine monoidal categories are related to laziness.

As is well known, we can program with monoidal categories using monads (e.g.~\S\ref{sec:monadsmonoidal}). The key conceptual contribution of this paper is the observation that we should program with two separate monads: an affine monad \lstinline|Prob| of probability measures, allowing laziness, and a non-affine monad \lstinline|Meas| of unnormalized measures, allowing the Bayesian Monte Carlo methods (\S\ref{sec:intro:mc}). This is the basis of our metalanguage in Section~\ref{sec:interface}. Theorem~\ref{thm:iid} demonstrates the way that lazy structures work with probability, in terms of infinite streams of samples. Proposition~\ref{prop:undecidable} shows that, once we have this functionality, we have to separate \lstinline|Prob| and \lstinline|Meas| to avoid deciding the halting problem. 

When we regard types in probabilistic programming languages as spaces of synthetic measure theory, we see spaces from non-parametrics, such as function spaces and infinite lists, behaving intuitively and straightforwardly, even though they can be subtle from the traditional measure-theoretic approach. We give a semantic model of our metalanguage in quasi-Borel spaces ( Section~\ref{sec:monads}), but a novelty here is to fix the basic probability space  to a space of lazy rose trees~(\S\ref{sec:msmh:rose}). This lends itself to an implementation in the form of our Haskell library LazyPPL~\cite{nonanongit}. 

\subsection{New Metropolis-Hastings-based inference algorithms}
\label{sec:intro:mhg}
To experiment with these examples involving laziness, we introduce new inference algorithms that build on earlier inference methods for probabilistic programming (e.g.~\cite{lightweight-mh}). These take as an argument a program describing an unnormalized measure, and produce a stream of samples as output.

Traditional uses of Monte Carlo algorithms often assume a finite-dimensional state space; sometimes they can adapt to changing dimensions (e.g.~\cite{rjmcmc}). But in a purely lazy setting, the distributions are implicitly infinite-dimensional, and indeed our basic probability space is the infinite dimensional space of rose trees. To resolve this, we provide new instantiations of the Metropolis-Hastings-Green algorithm, that do apply in this setting, and which we have also implemented in Haskell~\cite{nonanongit}.
\begin{itemize}
\item Our main algorithm (\S\ref{sec:msmh}) is purely lazy. It operates lazily over the entire infinite-dimensional state space, mutating different parts at random.
\item Our other algorithms work by mixing kernels (\S\ref{sec:single-site}). In particular, we are able to implement roughly the algorithm of~\cite{lightweight-mh}, adapted to this lazy setting, by using Haskell internals (\lstinline|ghc-heap|) to identify which dimensions are actually being used in a given run of the program.
\end{itemize}
We can show that these algorithms are correct (via Theorems~\ref{thm:mutatetrees-green}, \ref{th:mhg}, \ref{thm:state-dependent-mixing}).
Generally speaking, general purpose algorithms such as these will not work as efficiently as hand-crafted inference methods for specific scenarios. Nonetheless, they are useful for prototyping the numerous examples we consider in this paper to illustrate laziness and monads in probabilistic programming (\S\ref{sec:easyexamples}, \S\ref{sec:randfun}).

\paragraph{Acknowledgements.} We have benefited from many helpful discussions, including with Victor Blanchi, Reuben Cohn-Gordon, Cameron Freer, Ohad Kammar, Dan Roy, Adam \'Scibior, Matthijs V\'ak\'ar, Frank Wood, Hongseok Yang, Mathieu Huot and other colleagues in Oxford. Thanks also to anonymous reviewers. We also benefited from presenting aspects of this work at various venues, including ACT, AIPLANS, FSCD, LAFI, PROBPROG and an early version in taught courses in OPLSS 2019 and Oxford (BSPP 2020). 
Research supported by AFOSR award number FA9550-21-1-0038; the ERC BLAST grant; and a Royal Society University Research Fellowship.


%% file: interface.tex
\label{sec:interface}
The idea of Monte Carlo based inference is that we define an unnormalized measure,
by weighting different random choices, and then Monte Carlo inference provides samples from
the normalized form of this measure. 
 \begin{figure}\newcommand{\tj}[3]{#1\vdash \mlstinline{#2 :: #3}}\newcommand{\tc}[2]{\mlstinline{#1 :: #2}}
\[\boxed{\begin{aligned}   &
                     \text{\emph{Types: } }\mlstinline{a}, \mlstinline{b}\ \mathrel{{:}{:}=} \mlstinline{RealNum}~|~\mlstinline{()}~|~\mlstinline{(a,b)}~|~\mlstinline{a -> b}~|~\mlstinline{Prob a}~|~\mlstinline{Meas a}~|~\dots
\\
   &\text{\emph{Terms:} }\mlstinline{t}, \mlstinline{u}\ \mathrel{{:}{:}=}
\hide{\mlstinline{0}~|~
\mlstinline{1}~|~
\mlstinline{t*u}~|~
\mlstinline{()}~|~
\mlstinline{(t,u)}~|~}
\mlstinline{x}~|~
\mlstinline{\\x -> t}~|~
\mlstinline{t u}~|~
\mlstinline{do \{x <- t ; u\}}
~|~
\mlstinline{return t}~|~
\hide{\mlstinline{sample}~|~
\mlstinline{score}~|~}
\dots
\\
\ &\text{\emph{Typing judgement} }(\tj \Gamma t a)\text{\emph:}\\
&\ \begin{prooftree}\phantom{\Gamma\mlstinline{t}}-\phantom{\Gamma\mlstinline{t}}\justifies
    \tj {\Gamma,\mlstinline{x::a},\Gamma'} x a\end{prooftree}\qquad
\begin{prooftree}    \tj {\Gamma,\mlstinline{x::a}} t b\justifies
    \tj \Gamma {\\x -> t} {a -> b}\end{prooftree}\qquad
\begin{prooftree}    \tj \Gamma t {a -> b}\quad
    \tj \Gamma u {a}\justifies
    \tj \Gamma {t u} {b}
\end{prooftree}\ \\[3pt]
&\ \begin{prooftree}    \tj \Gamma t {a}\justifies
    \tj \Gamma {return t} {m a}\end{prooftree}\qquad
  \begin{prooftree}    \tj \Gamma t {m a}\quad
    \tj {\Gamma,\mlstinline{x::a}} u {m b}\justifies
    \tj \Gamma {do \{x <- t ; u\}} {m b}
\using \mlstinline{m}\in\{\mlstinline{Prob},\mlstinline{Meas}\}\end{prooftree}\\[3pt]
&\emph{Typed constants:}\\
 &\ 
    \tc {()} {()}\quad
    \tc {(,)} {a -> b -> (a,b)}\quad
    \tc {fst} {(a,b) -> a}\quad
    \tc {snd} {(a,b) -> b}\\&\ 
    \tc {0,1} {RealNum}\quad
    \tc * {RealNum -> RealNum -> RealNum}\quad\dots
\\&\ 
    \tc {sample} {Prob a -> Meas a}\quad
    \tc {score} {RealNum -> Meas ()}\quad
  \end{aligned}}
\]
\vspace{-5mm}
\caption{\label{fig:metalang}Summary of the types and terms of the monadic metalanguage.}
 \vspace{-5mm}
 \end{figure}
In this section, we encapsulate this in programming terms by using two monads, describing normalized and unnormalized measures. To make this formal, we set up an instance of Moggi's monadic metalanguage (\cite{moggi:computation_and_monads}) outlined in Figure~\ref{fig:metalang}. We discuss the syntax now, with a first example in Section~\ref{sec:easyexamples:reglin}. We then provide a basic equational theory (\S\ref{sec:monad-eq}) and use it to show that  
the separation between these monads is crucial for exploring lazy data structures (\S\ref{sec:afflazy}, Theorem~\ref{thm:iid} and Proposition~\ref{prop:undecidable}). Throughout the section, we consider simple extensions of the metalanguage that support lazy structures.
%

In the metalanguage, there is a distinguished type \lstinline|RealNum|, thought of as the real numbers, and there are two monads:
 \begin{itemize}
\item A probability monad \lstinline|Prob| (in green font) so that
  \lstinline|(Prob a)| intuitively contains probability measures on \lstinline|a|.
  We typically have stock probability measures, such as \lstinline|uniform :: Prob RealNum|
  and \lstinline|normal :: RealNum -> RealNum -> Prob RealNum|,
  but we do not need to postulate these at this point. 
  \item A measures monad \lstinline|Meas|  (in red font), so that
    \lstinline|(Meas a)| intuitively contains unnormalized measures on \lstinline|a|.
 \end{itemize}
  There are two key operations:
  \begin{itemize}
  \item \lstinline|sample :: Prob a -> Meas a|, which allows us to regard a probability measure as an unnormalized measure;
  \item \lstinline|score :: RealNum -> Meas ()|, which provides a measure with given weight on a single point; this is an unnormalized measure unless the weight is $1$.
  \end{itemize}
  The \lstinline|score| operation is often used with a probability density, to incorporate the likelihood in a Bayesian scenario.
  One can use all kinds of distributions for observations, using their densities. 
  For example, to incorporate a Bayesian observation of data point~\lstinline|x| from a normal distribution with mean \lstinline[mathescape]|$\cb \mu$| and
  standard deviation \lstinline[mathescape]|$\cb \sigma$|, we write
  \lstinline[mathescape]|score (normalPdf $\cb \mu$ $\;\cb \sigma$ x)|, where $({\mathtt{normalPdf}\; \cb \mu\;\sigma\;\mathtt{x}})=e^{-({\cb\mathtt x}-{\cb \mu})^2/(2{
      \cb\sigma})^2}/({\cb \sigma}\sqrt {2\pi})$.
We use a Haskell notation for the metalanguage; we have implemented the metalanguage in Haskell (see \S\ref{sec:monads:impl}, \S\ref{sec:msmh}, \cite{nonanongit}) and so all the examples can be run.
\subsection{First example: Bayesian linear regression}
  \label{sec:easyexamples:reglin}
 
We illustrate the metalanguage with a simple 1-dimensional Bayesian regression model. (See \S\ref{sec:easyexamples} for further illustrations.) The problem of regression is that we have some data points observed, and we want to know which function generated those points.
Bayesian regression does not produce one single `line of best fit', but rather a probability distribution over the functions that might have
generated the points.  We start with a fairly uninformative prior distribution over linear functions, incorporate the likelihood of the observations, and produce a posterior by Monte Carlo simulation.

In statistical notation, we might define the prior by writing 
\[a \sim \mathcal N\mathit{ormal}(0,3)\qquad b\sim \mathcal N\mathit{ormal}(0,3)\qquad f(x) = ax +b
\]
Here, the slope \lstinline|a| and intercept \lstinline|b| are both drawn from normal distributions. 

In the metalanguage (extended with mild syntactic sugar) we can write
\hide{\begin{lstlisting}
 linear :: Prob (RealNum -> RealNum)
 linear = do a <- normal 0 3 
             b <- normal 0 3 
             let f x = a*x + b
             return f
\end{lstlisting}}
\begin{lstlisting}[columns=flexible]
 linear :: Prob (RealNum -> RealNum)
 linear = do { a <- normal 0 3 ; b <- normal 0 3 ; let f x = a*x + b ; return f }
\end{lstlisting}
We do not want to assume that the data points are exactly colinear, and so we do not want observe the likelihood of the data points being exactly \lstinline|f x|. Rather, we use a likelihood for the points being normally distributed around \lstinline|f x|, for some small standard deviation $\sigma$.
The inference problem might be written in statistical notation as:
\[f\sim \mathcal{L}inear \qquad
  y_i\sim \mathcal{N}\mathit{ormal}(f(x_i),\sigma)
  \qquad
  \text{What is $P(f|(y_i=d_i)_i)$?}
\]
where $d_i$ are the observed values at $x_i$. 
To this end, we define a general purpose function \lstinline|regress|, which takes a standard deviation $\sigma$, a prior over the function space \lstinline|prior|, and a list of \lstinline|(x,y)| observations \lstinline|dataset|. For convenience, we add a type of lists, and routines for operating over lists. 
\begin{lstlisting}[columns=flexible]
regress :: RealNum -> Prob (a -> RealNum) -> [(a,RealNum)] -> Meas (a -> RealNum)
regress (*@$
\cb \sigma$@*) prior dataset =
  do {f <- sample prior; forM_ dataset (\(x,d) -> score (normalPdf (f x) (*@$\cb \sigma$@*) d));  return f}
\end{lstlisting}
 \begin{figure}
\includegraphics[scale=.7]{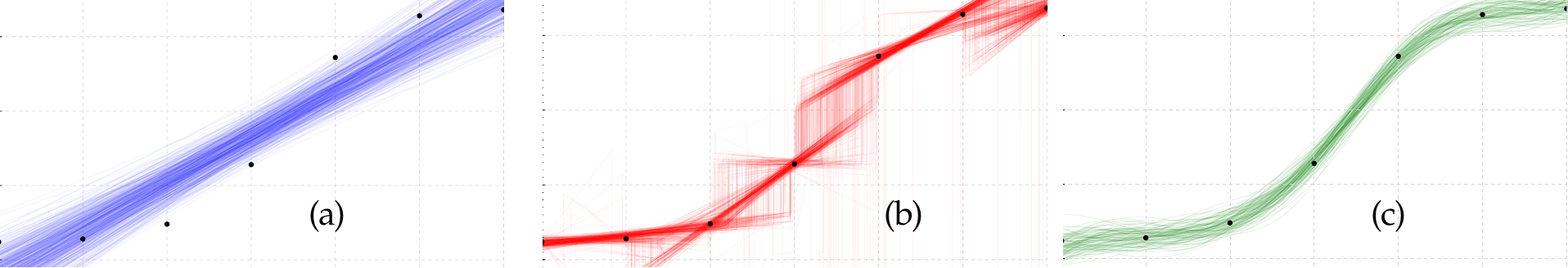}
  \caption{Bayesian regression in LazyPPL for the data set indicated by the dots. We illustrate the posteriors starting from three different priors on the function space. From left to right: (a)~linear (\S\ref{sec:easyexamples:reglin}), (b)~piecewise linear (\S\ref{sec:easyexamples:regpiecewiselin}), and (c) Gaussian processes (see \S\ref{sec:gp}).\label{fig:regression}}
\end{figure}
So linear regression in particular is achieved by performing Monte Carlo inference on the program \lstinline|(regress 0.1 linear dataset)|, see Figure~\ref{fig:regression} (a). We discuss inference in Section~\ref{sec:msmh}.

  \subsection{Equational reasoning}
\label{sec:monad-eq}

To begin to look more formally at the meaning of the programming syntax, we now consider equations between programs. We give a denotational model in Section~\ref{sec:monads}. One could also give an operational semantics, and we give an implementation in Section~\ref{sec:monads:impl}. Either way is slightly technical, so for now we work axiomatically, by listing some equations that we intend to hold. The idea is that the equations are sound, in that any programs that are derivably equal should be equal in any good denotational interpretation, or observationally equivalent for an operational semantics.

The equational theory is not intended to be a complete system, nor is it intended as a definitional semantics by rewriting, although the equations here are useful in practice, for example as compiler transformations.

We begin with the standard equational theory of the monadic metalanguage (Figure~\ref{fig:metalang-eq}), which is sound for strong monads on a cartesian closed category. Our interface has two monads, and we now state some additional equations that they should satisfy. 

First, we impose that \lstinline|sample :: Prob a -> Meas a| should be a monad morphism, i.e.:
\begin{equation}\label{eqn:score-morph}
  \begin{aligned}
  \mlstinline{sample (return t)} \quad&=\quad\mlstinline{return t :: Meas a} \\
  \mlstinline{sample (do \{x <- p ; q\})} \quad&=\quad\mlstinline{do \{x <- sample p ; sample q\} :: Meas b}
  \end{aligned}
\end{equation}
for \lstinline|t :: a|, \lstinline|p :: Prob a|, \lstinline|q :: Prob b|.

The measures monad should satisfy commutativity (e.g.~\cite{kock-comm-monad}): for
\lstinline|mx :: Meas a| and \lstinline|my :: Meas a|,
\begin{equation}\label{eqn:comm}
  \mlstinline{do \{x <- mx ; y <- my ; return (x,y)\}} \quad=\quad
  \mlstinline{do \{y <- my ; x <- mx ; return (x,y)\}}
\end{equation}
where $\mlstinline{x}\not\in\mathrm{freevars}(\mlstinline{my})$, $\mlstinline{y}\not\in\mathrm{freevars}(\mlstinline{mx})$.
The probability monad should satisfy commutativity, and also affinity (e.g.~\cite{kock-affine,jacobs-weakening}): 
\begin{equation}\label{eqn:affine}
\text{if \lstinline|mx :: Prob a|, then:} \quad \mlstinline{do \{x <- mx ; return ()\}} \qquad =\qquad \mlstinline{return ()}
\end{equation}
Finally, scores should combine multiplicatively, and a score of $0$ should trivialize:
\begin{equation}\label{eqn:scoremult}
  \begin{aligned}
&  \mlstinline{score 1} \ =\ \mlstinline{return () :: Meas ()}  \qquad \mlstinline{do \{score 0 ; mx\}} \ =\ \mlstinline{do \{score 0 ; my\} :: Meas a}\\
  &\mlstinline{do \{score r ; score s\}} \quad=\quad\mlstinline{score (r * s) :: Meas ()}
\end{aligned}
\end{equation}
for \lstinline|r,s :: RealNum| and \lstinline|mx,my :: Meas a|.

 \begin{figure}\newcommand{\eq}[2]{\mlstinline{#1} = \mlstinline{#2}}
   \[\boxed{\begin{aligned}   &
         \text{\emph{Typed equational theory:} we define a relation
           $\Gamma\vdash {\eq t u }{\mlstinline{::a}}$
           over typed terms $\Gamma\vdash { \mlstinline{t},\mlstinline{u} }{\mlstinline{::a}}$.}
         \\&\text{%
               Mostly the type and environment are clear from context;
               the exception is the unit axiom:}\\&\ \text{If }
             \Gamma \vdash\mlstinline{t :: ()}\text{ then }
             \Gamma\vdash \mlstinline{t}=\mlstinline{() :: ()}
\\
&\begin{array}{@{}ll}\text{Other axioms:}&\eq {do \{y <- do \{x <- s ; t\} ; u\}} {do \{x <- s ; do \{y <- t ; u\}\}}
                                           \quad [\mlstinline{x}\not\in\mathrm{fv}(\mlstinline{u})]
\\[2pt]&\eq {do \{x <- return t ; u\}} {u[t/x]}
      \qquad\eq {do \{x <- t ; return x\}} {t}\end{array}
         \\&\ 
      \eq{fst(t,u)}t\quad \eq{snd(t,u)}u\quad\eq t{(fst t,snd t)}
      \quad\eq{(\\x -> t) u} {t[u/x]}
      \quad\eq u{(\\x -> (u x))}
  \end{aligned}}
\]
\vspace{-5mm}
\caption{\label{fig:metalang-eq}Basic equational theory of the monadic metalanguage.
  We give additional equations in~\eqref{eqn:score-morph}--\eqref{eqn:stream}.}
 \end{figure}


  \subsection{Affine probability monads, laziness, and iid sequences}
\label{sec:afflazy}
We now show the connection between the affine monad law~\eqref{eqn:affine} and lazy data structures.
The reader might already glimpse a connection between~\eqref{eqn:affine} and laziness, since the law says that if the result \lstinline|x| of a program expression \lstinline|mx| is not used in what follows (left-hand side), then \lstinline|mx| need not be considered at all (right-hand side). 

To formally investigate this with lazy data structures,
we suppose that the metalanguage is extended with types \lstinline|Stream a| of streams,
  and standard accessor methods such as 
\[\mlstinline{hd :: Stream a -> a}\quad \mlstinline{tl :: Stream a -> Stream a}\quad
           \mlstinline{(:) :: a -> Stream a -> Stream a}\]
         together with syntactic sugar, and equations such as
         \begin{equation}\label{eqn:stream}\mlstinline{hd (x:xs)}=\mlstinline{x}\qquad
           \mlstinline{tl (x:xs)} = \mlstinline{xs}\qquad
           \mlstinline{(hd xs):(tl xs)} = \mlstinline{xs}
         \end{equation}
  We also suppose that we have a mechanism
  \lstinline|iid| for generating `independent and identically distributed' infinite random sequences.
  This should satisfy the following recursive equation:
  \begin{lstlisting}
    iid :: Prob RealNum -> Prob (Stream RealNum)
    iid p = do { x <- p ; xs <- iid p ; return (x : xs) }
  \end{lstlisting}
  We emphasise that the equation is for now regarded as an equational specification, not a recursive program definition with a given evaluation strategy. 
  Although the recursive call appears to make an infinite number of random choices, we can deduce from the basic equational laws that it must be treated lazily:
  \begin{theorem}\label{thm:iid}
    For any finite sequence \lstinline[mathescape]|i$\cb_1$|$,\dots,$ \lstinline[mathescape]|i$\cb_n$| of distinct indices,
    the following programs of type \lstinline|Prob (RealNum,...,RealNum)| are equal:
    \begin{enumerate}
    \item \lstinline[mathescape]|do {xs <- iid p ; return (xs !! i$\cb_1$,$\dots$, xs !! i$\cb_n$)}|
    \item \lstinline[mathescape]|do {x$\cb_1$ <- p ; $\dots$ ; x$\cb_n$ <- p ; return (x$\cb_1$ , $\dots$, x$\cb_n$)}|
    \end{enumerate}
  \end{theorem}
  \begin{proof}[Proof notes]
    For instance (and omitting \lstinline|do| for brevity): 
    \begin{align*}
      &\mlstinline{xs <- iid p ; return (xs !! 2,xs !! 1)}&&\text{[\lstinline|iid| equation, 3 times]}
      \\&=
      \mlstinline{x$\cb_0\;$ <- p ; x$\cb_1\;$ <- p ; x$\cb_2\;$ <- p ; xs <- iid p ; return ((x$\cb_0$:x$\cb_1$:x$\cb_2$:xs)!!2,(x$\cb_0$:x$\cb_1$:x$\cb_2$:xs)!!1)}\hspace{-5cm}
      \\&=
            \mlstinline{x$\cb_0\;$ <- p ; x$\cb_1\;$ <- p ; x$\cb_2\;$ <- p ; xs <- iid p ; return (x$\cb_2$,x$\cb_1$)}&&\text{[\eqref{eqn:stream}]}
      \\&=
      \mlstinline{x$\cb_1\;$ <- p ; x$\cb_2\;$ <- p ; return (x$\cb_2$,x$\cb_1$)}&&\text{[\eqref{eqn:affine}]}
      \\&=
      \mlstinline{x$\cb_1\;$ <- p ; x$\cb_2\;$ <- p ; return (x$\cb_1$,x$\cb_2$)}&&\text{[\eqref{eqn:comm} and rename variables]}\qedhere 
   \end{align*}
\end{proof}
  To emphasise the connection with lazy structures, we make the following remarks. Our equational theory has thus far not discussed `termination', however, in any reasonable semantics, program~(2) in Theorem~\ref{thm:iid} would be regarded as terminating. Thus program~(1), being equal to~(2), should also be regarded as terminating.
  Thus the call to \lstinline|iid| induces an unbounded stream of results, despite terminating, which is the essence of laziness.
  

Using similar arguments, we can assume a more general primitive that generates a stream of results by iterating a function \lstinline|b -> Prob (a, b)|. This is subject to the following  equation:  
\begin{lstlisting}
 unfold :: (b -> Prob (a, b)) -> b -> Prob (Stream a)
 unfold f y = do { (x, y') <- f y; xs <- unfold f y'; return (x : xs)} 
\end{lstlisting}
We note that we could define \lstinline|iid| by 
\lstinline|iid p = unfold (\_ -> do {x <- p; return (x, ())}) ()|.

\subsubsection{Discussion of problems with laziness in general measures monads}
We do not ask that the general measures monad, \lstinline|Meas a|, is affine, because then the \lstinline|score| construct would disappear. In practice, this means that we cannot expect to construct lazy data structures using recursive definitions in the
\lstinline|Meas| monad, as we now explain. 

Suppose for a moment that we did have the equivalent of
\lstinline|unfold| and \lstinline|iid| in the \lstinline|Meas| monad:
\begin{equation}
  \label{eqn:unfoldM}\begin{aligned}
&  \mlstinline{unfoldM :: (b -> Meas (a, b)) -> b -> Meas (Stream a)} 
  \\& \mlstinline{iidM :: Meas RealNum -> Meas (Stream RealNum)}\end{aligned}
\end{equation}
   For instance, \lstinline|iidM m = do { x <- m ; xs <- iidM m ; return (x : xs)}|. 
Consider the program
\lstinline|do { xs <- iidM (do {score 2 ; return 42}) ; return (hd xs)} :: Meas RealNum|. \mbox{}
  We are unable to use the equations of the monadic metalanguage to simplify this program, because we cannot use (\ref{eqn:affine}), and the semantics is unclear. (Possible semantics include: the program does not terminate, or, the program terminates with infinite score, since $2^n\to \infty$ as $n\to \infty$.)
  \hide{
    \\&=
  \mlstinline{score 2 ; x0 <- return 42 ; xs <- iidM p ; return ((x0 :: xs) !! 0)}
    \\&=
  \mlstinline{score 2 ; xs <- iidM p ; return 42 }
    \\&=
  \mlstinline{score 2 ; score 2 ; xs <- iidM p ; return 42}
    \\&=
  \mlstinline{score 2 ; score 2 ; score 2 ; xs <- iidM p ; return 42}
    \\&=
  \dots
}

It is still possible to use lazy data structures when building unnormalized measures:
for instance we can write \lstinline|sample (iid p) :: Meas (Stream a)| whenever \lstinline|p :: Prob a|. 
But we show now that implementing a general \lstinline|unfoldM| and \lstinline|iidM| such that
\lstinline|sample (iid p) = iidM (sample p)| is impossible, whenever our equations hold.   
\begin{proposition}\label{prop:undecidable}
  Suppose that there are terms \lstinline|unfoldM| and \lstinline|iidM| as in \eqref{eqn:unfoldM} such that we have
  \lstinline|iidM (return t) = sample (iid (return t))|.
  Then among the programs \[\{\mlstinline{t}~|~\exists \mlstinline{r},\mlstinline{u}\text{ such that } \mlstinline{t}=\mlstinline{do \{score r ; return u\} :: Meas RealNum}\}\]
    the score \lstinline|r| is undecidable.
\end{proposition}
\begin{proof}
  Pick an arbitrary Turing machine $M$.
  We exhibit a program \lstinline|h| such that we have either
  \lstinline|h = score 0 ; return 0| if $M$ halts, and \lstinline|h = score 1 ; return 42| otherwise. Thus it is in general undecidable whether the score is $0$ or $1$. 
To this end, write \lstinline|steps n| if $M$ halts after simulating it for \lstinline|n| steps. Suppose we extend the metalanguage with this \lstinline|steps| construction.
  \begin{lstlisting}[columns=flexible]
 h :: Meas Integer
 h = hd <$> (unfoldM (\n -> do{ score (if (steps n) then 0 else 1); return (42,n+1) })) 0
 \end{lstlisting}
 (where \lstinline|(t <$> u) = do {x <- u ; return (t x)}|). 
Suppose that $M$ terminates at time $n$. Then, unrolling \lstinline|<$>| and the recursive equation for \lstinline|unfoldM|, and simplifying using \eqref{eqn:score-morph}--\eqref{eqn:stream} and Fig.~\ref{fig:metalang-eq}:
  \begin{align*}\mlstinline{h}\ &=\; 
  \mlstinline{do \{xs <- (unfoldM $\dots$ ) 0 ; return (hd xs)\}}
\; =\;
  \mlstinline{do \{xs <- (unfoldM $\dots$ ) 1 ; return 42\}}\\ &= \quad\dots \quad =\;
  \mlstinline{do \{xs <- (unfoldM $\dots$ ) n ; return 42\}}
\\    
  &=\;
  \mlstinline{do \{score 0 ; xs <- (unfoldM $\dots$ ) (n+1) ; return 42\}}
\;=\;
  \mlstinline{do \{score 0 ; return 0\}}
   \text.
  \end{align*}
  On the other hand, if the machine $M$ does not halt,
  then \lstinline|h = hd <$> iidM (return 42)|.
  By the assumption in the theorem,
  \lstinline|h = hd <$> sample (iid (return 42))|,
  and so by (\ref{eqn:score-morph}) and (\ref{eqn:scoremult}), we have
  \begin{align*}
    \mlstinline{h}\ &=\;\mlstinline{hd <$\mbox{\textdollar}$> do \{xs <- sample (iid (return 42)) ; return (42 : xs) \}}&&\text{[\lstinline|iid| \& (\ref{eqn:score-morph})]}\\
    &=\; \mlstinline{do \{xs <- sample (iid (return 42)) ; return 42\}} &&\text{[$\mlstinline{<$\mbox{\textdollar}$>}$ \& (\ref{fig:metalang-eq})]}\\
    &=\; \mlstinline{sample (do \{xs <- iid (return 42) ; return 42\})} \;=\; \mlstinlin{sample (return 42)} &&\text{[(\ref{eqn:score-morph}) \& (\ref{eqn:affine})]}\\
    &=\; \mlstinline{return 42} \;=\; \mlstinline{do \{score 1 ; return 42\}} &&\text{[(\ref{eqn:score-morph}) \& (\ref{eqn:scoremult})]}\\\qedhere
  \end{align*}
  \end{proof}
If the score is undecidable, this means the system is not amenable to Monte Carlo simulation, because the score is needed to know whether or not to reject a run of the program. 
For the example in the proof of Proposition~\ref{prop:undecidable}, a Monte Carlo simulation would have to always reject if the machine~$M$ terminates, and always accept if the machine $M$ does not terminate, so it would have to decide the Halting problem, which is absurd.
There are other decidability problems in probabilistic programming (e.g.~\cite{aafrr-comput-cond-ind,hms-computable-ppl}), but this barrier occurs even without any statisticial primitives. 

    In summary, we can program random lazy data structures using the affine \lstinline|Prob| monad and constructions like \lstinline|iid|, and execution must treat them lazily (Theorem~\ref{thm:iid}). 
    By contrast, we cannot hope to program lazy data structures using the \lstinline|Meas| directly, for decidability reasons (Proposition~\ref{prop:undecidable}). 
    Nonetheless we can build random lazy data structures using the \lstinline|Prob| monad, and constructions such as \lstinline|iid|, and then regard them as arbitrary measures using \lstinline|sample :: Prob a -> Meas a|.

    \subsection{Comparison with prior work on synthetic measure theory}

    Previous work considered a synthetic measure theory based on a cartesian closed category with countable sums and products and a commutative monad $M$ that is countably additive~\cite{kock,scibior}. We can then construct an affine submonad $P\subseteq M$ as an equalizer, and an object of scalars $R = M(1)$.
    Any such model is a model of the metalanguage in Sections~\ref{sec:monad-eq}:
    $\mlstinline{Prob}=P$,
    $\mlstinline{Meas}=M$,
    $\mlstinline{RealNum}=R$,
    \lstinline|sample| is the inclusion, and \lstinline|score| is the identity function.

This connects to other axiomatizations too: for a model of synthetic measure theory, the Kleisli category of~$M$ is a CD/GS-category~\cite{cho-jacobs,fl-free-gs,dario-thesis}, and the Kleisli category of~$P$ is a Markov category~\cite{fritz,representable_markov}. Infinite processes such as \lstinline|iid| have recently been studied in Markov categories~\cite{infinite_products}, but not in the unnormalized setting.

One difference between our metalanguage and synthetic measure theory is that \lstinline|Meas|, \lstinline|Prob| and \lstinline|RealNum| are given as explicit ingredients here; \lstinline|Prob| is not necessarily an equalizer and \lstinline|RealNum| is not necessarily equal to \lstinline|Meas ()|. We found it helpful to make these distinctions in practice. For instance, the expression \lstinline|do {x <- sample uniform ; score (f x)} :: Meas ()| intuitively corresponds to the \lstinline|RealNum| that is the integral $\int_0^1 \mlstinline{f}(x)\,\dd x$, but for general~\lstinline|f|, this can only be calculated approximately, e.g.~via Monte Carlo simulation.


%% file: regression.tex
\label{sec:easyexamples}

We now give examples of probabilistic programs with lazy data structures, using the metalanguage from Section~\ref{sec:interface}. We are able to run the programs by using our implementation (LazyPPL, \S\ref{sec:monads:impl}, \cite{nonanongit}). 

The examples show that laziness, and two separate monads \lstinline|Prob| and \lstinline|Meas|, allow compositional programming; without this, we cannot use \lstinline|iid|, and we have to pass around truncation bounds.

\subsection{Point processes as lazy streams of points}\label{sec:poisson}
A random collection of points is called a \emph{point process}. A first example of an infinite point process is a \emph{Poisson} point process on the positive reals
(see also Figure~\ref{fig:affine} (a)).
In statistical notation, we consider a sequence of real random variables $x_0,x_1,\dots$ such that the gaps are exponentially distributed:
\[
  x_0\sim \mathcal E\mathit{xponential}(r)\qquad\quad
  (x_{i+1}-x_i)\sim \mathcal E\mathit{xponential}(r)
\]
This can be defined in a minor extension of our metalanguage, by giving a random stream
\footnote{Henceforth we use \lstinline|[a]| instead of \lstinline|(Stream a)|, recalling that Haskell lists include infinite lists.}:
\begin{lstlisting}[columns=flexible]
  poissonPP :: RealNum -> Prob [RealNum]
  poissonPP rate = do { steps <- iid (exponential rate) ; return (scanl (+) 0 steps) } 
       \end{lstlisting}

  This returns an infinite random stream, illustrated in Figure~\ref{fig:affine} (a). Although this is simple, it becomes powerful when we compose with other models, as we show in the following sections. (See also \cite{swaraj}, for more on programming with point processes.)

\paragraph{Comments on avoiding lazy data structures}       Without using infinite lazy data types, we can only implement a truncated Poisson point process, where an upper bound is specified:
\hide{\begin{lstlisting}
  poissonPPbounded :: RealNum -> RealNum -> RealNum -> Prob [RealNum]
  poissonPPbounded rate lower upper = 
      do step <- exponential rate
         let x = lower + step
         if x > upper then return [] else do
            xs <- poissonPPbounded rate x upper
            return (x : xs)
\end{lstlisting}}
\begin{lstlisting}[columns=flexible]
poissonPPbounded :: RealNum -> RealNum -> RealNum -> Prob [RealNum]
poissonPPbounded rate lower upper = do 
  step <- exponential rate
  let x = lower + step
  if x > upper then return [] 
  else do {xs <- poissonPPbounded rate x upper; return (x:xs)}
\end{lstlisting}
This truncated form is an obstacle to compositional modelling, because when the Poisson point process is used as part of a more complex model, we will have to calculate an upper bound and pass it around manually, as we will discuss further in Section~\ref{sec:easyexamples:reg}.

(As an aside, we mention that although the bounded program will always return a finite list in practice, this relies on a mathematical argument: any sequence of draws from an exponential distribution will almost surely go above a given upper bound in a finite number of steps. This is no problem in a general purpose programming language, but may prove subtle in, say, a dependently typed language with termination guarantees.)

Of course, some truncation must happen at some point. In practice, in LazyPPL, it happens at the top level, the plotting routine stops looking at the stream beyond the viewport, and the laziness propagates from here. 


\subsection{Piecewise regression with lazy change points}

\label{sec:easyexamples:reg}
\label{sec:easyexamples:regpiecewiselin}
\label{sec:reg}

We combine the Poisson Point Process (\S\ref{sec:poisson}) with the linear regression model (\S\ref{sec:easyexamples:reglin}), to obtain \emph{piecewise} linear regression, where the prior is over piecewise linear functions.   
We define a function that will splice together different draws from a random function given a random sequence of change points: 
\begin{lstlisting}[columns=flexible]
 spliceProb :: Prob [RealNum] -> Prob (RealNum -> RealNum) -> Prob (RealNum -> RealNum)
 spliceProb pp f = do {xs <- pp ; fs <- iid f ; return $ splice xs fs }
\end{lstlisting}
Here we assume a basic deterministic function
\begin{lstlisting}[columns=flexible]
 splice :: [RealNum] -> [RealNum -> RealNum] -> (RealNum -> RealNum)
\end{lstlisting}
which splices together a sequence of functions at the given change points. 
The sequence of change points can be infinite, so the function can have an infinite number of pieces.
But this is no problem if our probability monad is affine, and this can be handled lazily.

We can perform piecewise linear regression using a Poisson point process, for example via
\begin{lstlisting}[columns=flexible]
  piecewiseLinear :: Prob (RealNum -> RealNum)
  piecewiseLinear = spliceProb (poissonPP 0.2) linear
\end{lstlisting}
and then \lstinline{regress 0.1 piecewiseLinear dataset}, which gives Fig.~\ref{fig:regression} (b), using the inference in \S\ref{sec:msmh}.

\paragraph{Comments on avoiding lazy data structures.}
Our piecewise linear model demonstrates why it is important to have separate \lstinline|Prob| and \lstinline|Meas| monads. As we have explained in Section~\ref{sec:afflazy}, it is impossible to have  a general \lstinline|iid| on arbitrary measures, but it is fine on probability measures. We have used \lstinline|iid| twice in this model. It is part of \lstinline|spliceProb|, and so it is crucial to know that the prior \lstinline|f| is a probability measure on functions.

This piecewise linear regression can be implemented without infinite lazy structures, by using
\lstinline|poissonPPbounded| and rewriting \lstinline|spliceProb| to just sample a bounded number of functions \lstinline|fs|, depending on the length of \lstinline|xs|. For this first example, the bound can come from the chosen viewport, or data range, which might not be hard to calculate, but in general, propagating this is not compositional. 
For example, we can quickly rescale a random function on the $x$-axis:
\begin{lstlisting}[columns=flexible]
  rescale :: Prob (RealNum -> RealNum) -> Prob  (RealNum -> RealNum) 
  rescale p = do { f <- p ; return (\x -> f(2 * x)) }
\end{lstlisting}
We can then perform regression
with \lstinline{regress 0.1 (rescale piecewiseLinear) dataset}.
To do this without lazy structures, with explicit bounds, we would also need to rescale the bounds by hand. So compositionality with explicit bounds is impossible.
  


As an aside, we point out the beauty of higher-order probabilistic programming for compositional model building: it is very easy to switch the Poisson process for a different point process, or to use piecewise \emph{constant} regression, and so on (see the repository~\cite{nonanongit} for more examples).

\subsection{Clustering using a lazily broken stick}
\label{sec:dp}

For a set of data points, clustering is the problem of finding the most appropriate partition into clusters.
In non-parametric clustering, the number of clusters is unknown and unbounded.
As we demonstrate, this is closely connected to laziness. The separation into two monads,
\lstinline|Prob| and \lstinline|Meas|, is crucial for a compositional specification in our metalanguage. 

%

\subsubsection{Stick-breaking}
\label{sec:easyexamples:stickbreaking}
In \emph{Bayesian} clustering, one considers a prior distribution over possible partitions of the data into clusters. This is usually described in terms of stick-breaking, where the unit interval $[0,1]$ is broken into an infinite number of sticks, each representing a cluster, and the size of the stick is the proportion of points in that cluster. 

As an example, we consider stick-breaking for a Dirichlet process (e.g.~\cite{gv-bnp}), where at each step we break off a portion of the remaining interval according to a beta distribution.
 This is often written in statistical notation as
\[
r_k\sim\mathcal B\mathit{eta}(1,\alpha)\qquad
\textstyle v_k=r_k\cdot\prod_{i=1}^{k-1}(1-r_i)
\]
The random sequence $(v_0,v_1,\dots)$ has the property that 
$\sum_{i=0}^\infty v_i=1$, almost surely. 

 Stick-breaking is easy to define as a random stream \mlstinline{stickBreaking $\ \alpha$} in our metalanguage.
\begin{figure}
\includegraphics[scale=0.3]{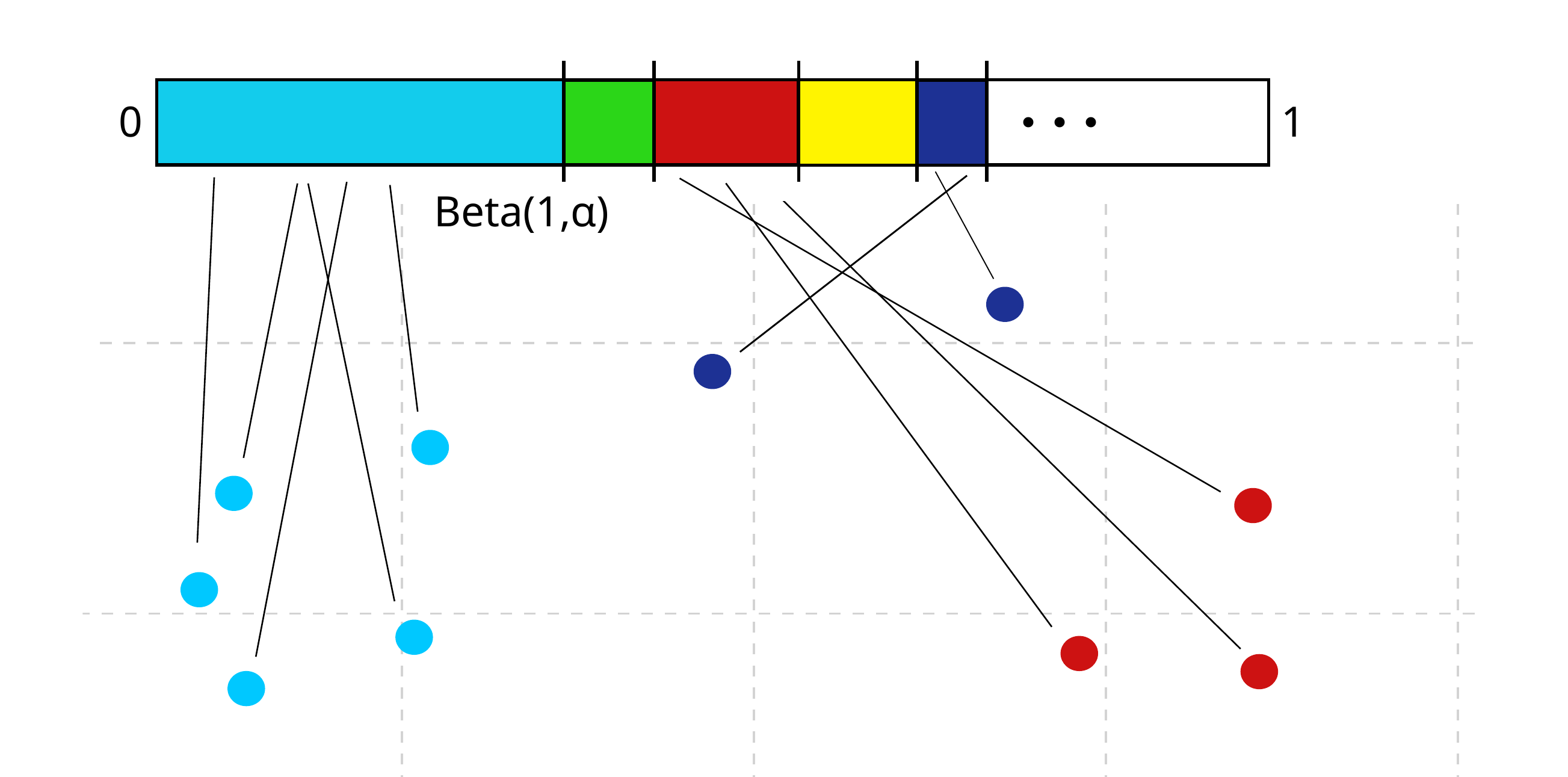}\qquad
\includegraphics[scale=0.18]{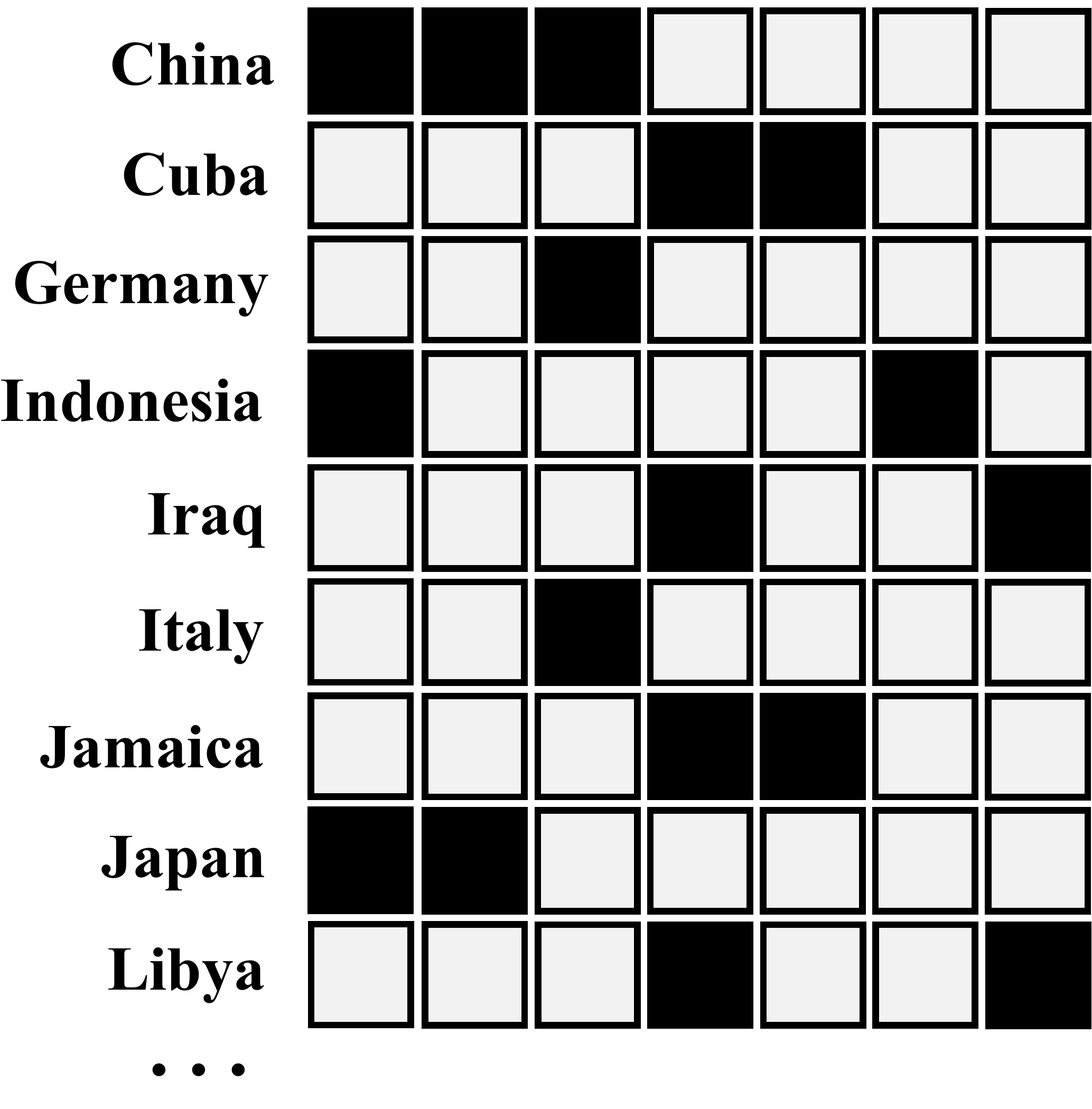}
\caption{From left to right: (a)~Dirichlet process clustering by stick-breaking (\S\ref{sec:dp});
  (b)~\label{fig:ibp}Feature extraction for a psychological experiment (\S\ref{sec:ibp}): the features (columns) are automatically inferred and often interpretable.\label{fig:dp}}
\end{figure}
\hide{\begin{lstlisting}
 stickBreaking :: RealNum -> Prob [RealNum]
 stickBreaking (*@$\cb\alpha$@*) = 
   do rs <- iid $ beta 1 (*@$\cb\alpha$@*)
      let vs = zipWith (*) rs $ scanl (*) 1 $ map (1-) rs
      return vs
\end{lstlisting}}
\begin{lstlisting}[columns=flexible]
 stickBreaking :: RealNum -> Prob [RealNum]
 stickBreaking (*@$\cb\alpha$@*) = do rs <- iid $ beta 1 (*@$\cb\alpha$@*)
                         let vs = zipWith (*) rs $ scanl (*) 1 $ map (1-) rs
                         return vs
\end{lstlisting}
\hide{\begin{lstlisting}[columns=flexible]
 stickBreaking :: RealNum -> Prob [RealNum]
 stickBreaking (*@$\cb\alpha$@*) = 
   do { rs <- iid (beta 1 (*@$\cb\alpha$@*)); return (zipWith (*) rs $ scanl (*) 1 $ map (1-) rs)}
\end{lstlisting}}
We can then regard this sequence of probabilities, which sums to $1$, as a distribution on natural numbers.
One way to do this is to use a uniform sample in $[0,1]$ to index into the sequence, via elementary functional programming. 
\begin{lstlisting}[columns=flexible]
 toProb :: [RealNum] -> Prob Int
 toProb vs = do {r <- uniform; return $ findIndex (> r) (scanl (+) 0 vs)} 
\end{lstlisting}
For clustering, we also assign some parameters to each cluster, i.e.~to each part of the stick. 
For example, if the clusters are located geometrically, we might associate a random position value to each cluster, assuming that the data points in this cluster follow a normal distribution around that position.
The parameters are taken from a `base' measure. The general procedure is the Dirichlet Process:
\hide{\begin{lstlisting}
 dp :: RealNum -> Prob a -> Prob (Prob a)
 dp (*@$\cb\alpha$@*) base = 
   do vs <- stickBreaking
      xs <- iid base
      return $ do {n <- toProb vs ; return $ xs !! n} 
\end{lstlisting}}
\begin{lstlisting}[columns=flexible]
dp :: RealNum -> Prob a -> Prob (Prob a)
dp (*@$\cb\alpha$@*) base = do {vs <- stickBreaking (*@$\cb\alpha$@*); 
  xs <- iid base; return $ do {n <- toProb vs; return (xs!!n)}}
\end{lstlisting}
%
%
We can then piece together a Dirichlet Process clustering model. 
\begin{lstlisting}[columns=flexible]
 cluster :: RealNum -> (Prob a) -> ((b, a) -> RealNum) -> [b] -> Meas [(b,a)]
 cluster (*@$\cb\alpha$@*) base likelihood dataset = do
  p <- sample $ dp (*@$\cb\alpha$@*) base
  xs <- sample $ iid p
  let taggedData = zip dataset xs 
  mapM_ (score . likelihood) taggedData
  return taggedData
\end{lstlisting}
This model (\mlstinline{cluster $\ \cb\alpha\ $ base likelihood dataset}) partitions the data into random clusters via stick-breaking, draws a value from \lstinline|base| for each cluster, and incorporates a score for each data point using \lstinline|likelihood|.

\paragraph{Comments on avoiding lazy data structures}
Our clustering model demonstrates why it is important to have a separate \lstinline|Prob| monad. 
We have used \lstinline|iid| three times in this model. As previously mentioned, it is impossible to have  a general \lstinline|iid| on arbitrary measures, but it is fine on probability measures (\S\ref{sec:afflazy}). For example, \lstinline|iid| is used in \lstinline|dp|, and so it is crucial to know that the argument \lstinline|base| is a \emph{probability} measure.

In the literature, there are various ways of avoiding the infinite lazy data structure of stick breaking. 
One common approach is to truncate the distribution. For example, we might pick some small $\epsilon$, and then find $n$ such that $(1-\sum_{i=1}^n v_i) < \epsilon$, so that the chance of needing $v_i$ with $i\geq n$ is small enough to be ignored. Care needs to be taken, then, to ensure that this $\epsilon$ is chosen suitably for the model, which is potentially difficult and non-compositional when stick-breaking is part of a complex model.

\subsection{Non-parametric feature assignment}
\label{sec:ibp}
Clustering assigns a single cluster to each data point. We briefly also mention the task of \emph{feature assignment}, which  assigns a finite set of features to each data point. A feature assignment for a data set consists of a finite set of features, together with a subset of these features for each data point. For example we can have a set of movies, and a feature assignment could be a set of genres for each movie. In the non-parametric setting, the number of features is unbounded even for a fixed data set.

As an application, we consider a basic problem from applied psychology \cite{navarroNonparametricBayesianMethod2006}. The problem is as follows. We have a set of countries together with a similarity coefficient for each pair of distinct countries, calculated based on answers from participants in a study.  
The goal is to infer a set of underlying features which characterize the countries in the participants' minds, and influence their judgement. 
The outcome is shown in Figure~\ref{fig:ibp} (b), as a Boolean matrix indicating assignments of features (columns) to countries (rows).
Notice that the inferred features do match actual perceived features, such as continent (column~1), industry (col.~2), development (col.~4), conflict (col.~7).  

We now describe the model, and where the laziness comes in. 
The overall process is described as a function 
\begin{lstlisting}[columns=flexible]
   bp :: RealNum -> Prob a -> Prob (Prob [a])
\end{lstlisting}
that corresponds to a combination of Beta and Bernoulli processes (e.g.~\cite{thibaux-jordan}).
This generalizes the Dirichlet process \lstinline[breaklines=true]|dp :: RealNum -> Prob a -> Prob (Prob a)| (\S\ref{sec:dp}) by producing a list rather than a single value. See~\cite{nonanongit} for details.

\paragraph{A note about implementation}
The Beta process can be constructed using a variant of stick-breaking: each stick now corresponds to a feature and its size is the proportion of points with that feature \cite{BetaStickBreaking}. But in practice this construction can only be used as an approximation, because we can only look at a finite number of sticks and the features are unbounded. By contrast, for the Dirichlet process we can stop stick-breaking once we have found a cluster. These difficulties appear to coincide with continuity issues in the analysis of the Beta process (e.g.~\cite[Sl.~6]{roy-ibp}, also~\cite{Roy-ChurchNP,roy2014continuum,aafrr-comput-cond-ind}). 

In our implementation~\cite{nonanongit} we avoid these issues by making use of some amount of hidden state. We set up an Indian Buffet process \cite{ibp}, which is a way to assign features to data points sequentially by keeping track of features assigned to previous points. This terminates, and gives the same model~\cite{thibaux-jordan}. The state is safely encapsulated (see e.g.~\cite{xrp-roy}), and we use streams to deal with the unbounded aspects.

\section{Further examples with random functions}
\label{sec:randfun}
Our examples so far have focused on random streams (\S\ref{sec:easyexamples}) and the \lstinline|iid| construction (\S\ref{sec:afflazy}). 
The type of streams is often isomorphic to the function type \lstinline|Nat -> a|, regarding a stream as the function that returns its value at each index. 
In this section we briefly discuss the role of the affinity of the monad~\lstinline|Prob| in exploring function types. 

We have seen random functions \lstinline|linear| and \lstinline|piecewiseLinear :: Prob (RealNum -> RealNum)| in Section~\ref{sec:easyexamples}. We now consider Gaussian Processes (\S\ref{sec:gp}) and stochastic memoization (\S\ref{sec:easyexamples:stochmem}), as other ways of building random functions. There are plenty more besides; in the repository we define random functions based on random programs, so that regression becomes program induction~\cite{nonanongit}.

\subsection{Gaussian process regression}
\label{sec:gp}

A Gaussian process (e.g.~\cite[Ch.~11]{gv-bnp}) is a random function $f$ with the property that for a finite input sequence $(x_1\dots x_n)$ the output distribution $(f(x_1)\dots f(x_n))$ is a multivariate normal distribution.
For simplicity, we focus on 1-dimensional Gaussian processes, which are random functions $\RR\to\RR$.
As with any function, we can call these with an unbounded number of arguments: $n$ is not fixed.

A Gaussian process is parameterized by a mean function $m:\RR\to \RR$ and a covariance
function $k:\RR^2\to \RR$. Gaussian processes admit conditioning: in statistical notation, if 
$f\sim \mathcal{GP}(m,k)$, then 
\[\text{for all $x_0\in \RR$, }\quad f(x_0)\sim \mathcal{N}\mathit{ormal}(m(x_0),\sqrt{k(x_0,x_0)})
    \quad\text{ and }\quad f|(x_0\mapsto y)\sim \mathcal{GP}(m|_{x_0\mapsto y},k|_{x_0\mapsto y})
  \]
  where $m|_{x_0\mapsto y},k|_{x_0\mapsto y}$ are respectively the conditional mean and covariance functions\footnote{$m|_{x_0\mapsto y}(x) = m(x) + \frac {k(x_0,x) (y - m(x_0))}{k(x_0,x_0)}$, and 
    $k|_{x_0\mapsto y}(u,v)=k(u,v)-\frac{k(u,x_0)k(x_0,v)}{k(x_0,x_0)}$.}.
To formalize this in our metalanguage, we extend it with a Gaussian process primitive
\begin{lstlisting}[columns=flexible]
gp :: (RealNum ->RealNum) -> (RealNum ->RealNum ->RealNum) -> Prob (RealNum ->RealNum)
\end{lstlisting}
and consider the following equation 
\hide{if $f\sim \mathcal{GP}(\mathit{cov})$
\[
  (f(x_1),\dots ,f(x_n))\sim \mathcal N\mathit{ormal}(0, \Sigma^{\mathit{cov}}_{x_1,\dots,x_n})
\]
where $\Sigma^{\mathit{cov}}_{x_1,\dots,x_n}$ is a covariance matrix, $\Big(\Sigma^{\mathit{cov}}_{x_1,\dots,x_n}\Big)_{i,j}=\mathit{cov}(x_i,x_j)$.}%
at type \lstinline|Prob (RealNum -> RealNum)|, for any \lstinline[mathescape]|x$\cb_0$|:
\begin{equation}\label{eqn:gp}
\begin{aligned}
     \mlstinline{gp m k}\ \ 
=\quad \mlstinline{do }\ \ &
\mlstinline{y <- normal (m x$\cb_0$) (sqrt (k x$\cb_0\;$ x$\cb_0$))}
\\&
\mlstinline{f <- gp (m' x$\cb_0\;$ y) (k' x$\cb_0\;$ y) }
\\&
\mlstinline{return (\\x -> if x==x$\cb_0\;$ then y else f x)}
\end{aligned}\end{equation}
where \mlstinline{m' x$\cb_0\;$ y} and \mlstinline{k' x$\cb_0\;$ y} are the conditional mean and covariance functions.
To lazily evaluate a Gaussian process we can use~\eqref{eqn:gp} to unroll it as necessary, using different values for~\lstinline[mathescape]|x$\cb_0$|, and then use
the affine property of the probability monad \eqref{eqn:affine} to remove any reference to \lstinline|gp|
(c.f.~Theorem~\ref{thm:iid}).

Once \lstinline|gp| is in our metalanguage, we can combine Gaussian processes with other random functions. 
For example, we can use a random linear function (\S\ref{sec:easyexamples:reglin}) as a mean for a Gaussian process:
\begin{lstlisting}[columns=flexible]
 gpWithLinear:: Prob (RealNum -> RealNum)
 gpWithLinear = do { f <- linear ; return $ gp f (rbf 1 1) }
\end{lstlisting}
(using the standard radial basis kernel \lstinline|rbf|). 
The result of running \lstinline|regress 0.3 gpWithLinear dataset| is a posterior over smooth curves fitting the points, shown in Figure~\ref{fig:regression} (c). 

\paragraph{Note about implementation}
Although the defining equation~\eqref{eqn:gp} is straightforward and reminiscent of Theorem~\ref{thm:iid}, it is different from the other distributions so far in that it is not a Haskell-style definition.
We implemented \lstinline|gp| in LazyPPL~\cite{nonanongit} by first making an infinite sequence of
standard normals, using \lstinline|iid (normal 0 1)|, and 
then by using a hidden memo table and standard linear algebra for conditional probability in Gaussian processes. 
In general, hidden state would violate the affine and commutativity properties; in this setting, since we validate~\eqref{eqn:gp}, the affine and commutative properties of \lstinline|Prob| (\ref{eqn:comm},\ref{eqn:affine})  remain satisfied.

\paragraph{Comments on the role of lazy structures}
Our Gaussian process regression example terminates because \eqref{eqn:gp} allows us to regard the value of the function \lstinline|f| at certain inputs, and disregard its value at unused arguments. The input values are the points of the data observations, and the points used for plotting, which are dependent on the viewport and the resolution. 

Since \lstinline|gp| uses the affine \lstinline|Prob| monad, and not the \lstinline|Meas| monad, it can be passed to the second-order distributions above from Section~\ref{sec:easyexamples}. We can write \lstinline|spliceProb (poissonPP r) (gp m k)| to make a piecewise Gaussian process, which amounts to what is called a `jump process' in statistics. We can write \lstinline[mathescape]|dp $\cb\alpha$ (gp m k)| for a Dirichlet process mixture of Gaussian processes. 

The function space of the metalanguage is here being used to hide a lazy process. One \emph{could} replace the call to \lstinline|gp| with a multivariate normal distribution of fixed dimension, but then the plotting points and observation points would need to be passed manually as an input to the function \lstinline|gpWithLinear|. The dimension would need to be picked based on the number of data points and the plotting resolution.

With the lazy behaviour of \lstinline|gpWithLinear|, we can easily rescale the function, for example using 
\lstinline|rescale gpWithLinear :: Prob(RealNum -> RealNum)|.
(We can consider more exotic rescalings too, such as `deep Gaussian processes'~\cite{deep-gp}, by composing Gaussian processes.)
With a non-lazy version based on explicit multivariate normal distributions, one would also need to take care to manually rescale the viewport, resolution, and the observation points, which would not be compositional. 


\subsection{Stochastic memoization}
\label{sec:easyexamples:stochmem}

In standard programming, memoization is a form of laziness where a function caches previous results instead of re-calculating~\cite{memo-nature}. In functional probabilistic programming, memoization becomes a powerful method for building infinite-dimensional probability measures (e.g. ~\cite{goodman:church,Roy-ChurchNP,wood-sequence-memoizer}).
We now discuss memoization in the setting of the metalanguage, with the separation of \lstinline|Prob| and~\lstinline|Meas|. 

In statistical terms, suppose we have a parameterized distribution, $\mathcal P(x)$ over a space $B$, with
parameters $x$ in a space $A$. 
Stochastic memoization provides a random function $F:A\to B$ with
\[
\forall x\in A.\quad  F(x)\ \sim\ \mathcal P(x)\text.
\]
We can study stochastic memoization in our metalanguage by adding a constant
\begin{lstlisting}
  memoize :: (a -> Prob b) -> Prob (a -> b)
\end{lstlisting}
The idea is that \lstinline|memoize p :: Prob (a -> b)|
randomly picks an assignment of results for each argument, informally by sampling once from \lstinline|(p x)| for every \lstinline|(x :: a)|.
 We consider the following equation at type \lstinline|Prob (a -> b)|:
\begin{equation}\label{eqn:mem}
    \hspace{-3mm} \mlstinline{memoize p}
= \ 
\mlstinline{ do \{ y <- p x$\cb_0$; f <- memoize p; return (\\x -> if x==x$\cb_0\;$ then y else f x)\}}\hspace{-3mm}
\end{equation}
As with \lstinline|gp|, if the memoized function is only called with a finite collection of arguments,
we can unroll equation~\eqref{eqn:mem} for each of the arguments, and then use the affine law~(\ref{eqn:affine}) to remove any reference to \lstinline|memoize|.


\subsubsection{Memoization on countable types}
By regarding \lstinline|Nat -> a| as isomorphic to \lstinline|Stream a|, we can define memoize with domain type \lstinline|Nat| 
just using basic primitives for lazy monadic streams, according to the following equation:
\begin{lstlisting}[columns=flexible]
 memoize :: (Nat -> Prob b) -> Prob (Nat -> b)
 memoize f = do { ys <- mapM f [0..] ; return $ \x -> ys !! x }
\end{lstlisting}
(In practice~\cite{nonanongit}, we define \lstinline|memoize| more efficiently using tries; c.f.~\cite{Hinze00memofunctions}.)
We can see, in particular, that
\lstinline|memoize (\_ -> p)| corresponds to \lstinline|iid p|. 
With this in mind, following Proposition~\ref{prop:undecidable}, we cannot hope to have
\lstinline|memoizeMeas :: (a -> Meas b) -> Meas (a -> b)|, and the restriction to the probability type is crucial.

\subsubsection{Illustrations of using memoization}

In Section~\ref{sec:dp} we have seen the Dirichlet process applied to clustering.
One common method in statistics is to regard the Dirichlet process as a random measure, by applying it to a generic base measure, for example,
\begin{lstlisting}
  dp (*@$\cb \alpha$@*) uniform :: Prob (Prob RealNum)
\end{lstlisting}
That is to say, rather than assigning informative parameters to the clusters (such as mean position) from the outset, we begin by assigning each cluster a name, picked uniformly from $[0,1]$.

Later on, of course, we would like to assign a position in (say) 2D space to each cluster. Since the clusters are named by real numbers, this assignment of positions to names of clusters is a random function $f:\RR\to \RR^2$ such that, say,
\[
  \forall r.\,   f_1(r) \sim \mathcal{N}\mathit{ormal}(0,3)\quad\&\quad
  f_2(r) \sim \mathcal{N}\mathit{ormal}(0,3)
  \text.
\]
This can be implemented in our metalanguage by using memoization:
\begin{lstlisting}[columns=flexible]
 p <- dp (*@$\cb \alpha$@*) uniform; xpos <- memoize (\_ -> normal 0 3); ypos <- memoize (\_ -> normal 0 3)
\end{lstlisting}
If we return \lstinline|do {r <- p ; return (xpos r,ypos r)}|, then this whole 
program is equivalent to using a different base distribution,
\lstinline[mathescape]|dp $\cb \alpha$ (mvnormal (0,0) (3,3))|.
However, the memoized form is more compositional: we can later assign other attributes to clusters in another part of the model (see~\cite{nonanongit}).



\hide{
In the clustering example, we can use memoization to assign a mean to each cluster. One implementation of our clustering model above is:
\begin{lstlisting}
 cluster dataset base likelihood = do 
    assignment <- sample $ assignClusters 0.3 dataset 
    positions <- sample $ memoize $ \c -> base -- lazily map values to clusters
    mapM (\(x, c) -> score $ likelihood (positions c) x) assignment 
    return assignment
\end{lstlisting}
The memoized function \lstinline|positions| is an infinite-dimensional object. There is no bound on the number of clusters, but only the dimensions needed need be computed. }



%% file: monads.tex
\label{sec:monads}
We now provide a model of the metalanguage from Section~\ref{sec:interface}, with two monads:
\lstinline|Prob| and \lstinline|Meas|. 
Our model is based on quasi-Borel spaces~\cite{qbs}. The novelty here is that, in prior work on quasi-Borel spaces, one works up to lots of measurable isomorphisms, e.g.~$\RR\cong\RR^\NN$. Here we are explicit about this, and so the connection to laziness is clarified. Because we are working more concretely, a new implementation is suggested (\S\ref{sec:monads:impl}, \cite{nonanongit}). 

\subsection{Affine and commutative monads from monoidal categories}\label{sec:monadsmonoidal}
\newcommand{\CatA}{\mathcal V}
\newcommand{\CatB}{\mathcal C}
\newcommand{\ob}{\mathit{ob}}

We begin by recalling a convenient presentation of commutative monads based on monoidal categories.
Recall that a symmetric monoidal category is a category $\CatB$ with a functor
${\otimes}:\CatB\times\CatB\to \CatB$ and an object $I$, together with coherent associativity and symmetry data.
It is called \emph{affine} if $I$ is a terminal object. The following is well known in some circles (e.g.~\cite{freyd-cat}), and allows us to define commutative monads by directly defining their Kleisli categories. 
\begin{lemma}\label{lem:freyd-cat}
  Let $J:\CatA\to \CatB$ be a functor that is identity on objects (i.e.~$\ob(\CatA)=\ob(\CatB)$).
  \begin{enumerate}
  \item If $J$ has a right adjoint $R$ then $\CatB$ is the Kleisli category of the monad $RJ$.
  \item Suppose moreover that $\CatA$ has products.
    Then to give $RJ$ the structure of a commutative monad~(\ref{eqn:comm}) is to give $\CatB$ the structure of
    a symmetric monoidal category such that $J$ is a strong symmetric monoidal functor.
  \item Moreover, $RJ$ is an affine monad~(\ref{eqn:affine}) if and only if $\CatB$ is affine monoidal.
  \end{enumerate}
\end{lemma}
Suppose that $\CatB$ is the Kleisli category for a monad, say \lstinline|Meas|. 
The composition in $\CatB$ corresponds to the sequencing of the metalanguage.
First, a context $\Gamma=(\mlstinline{x$\cb_1$::a$\cb_1$,...,x$\cb_n$::a$\cb_n$})$ is interpreted as the monoidal product
$\mlstinline{a$\cb_1$}\otimes\dots \otimes \mlstinline{a$\cb_n$}$ in $\CatB$. 
Then a term $\Gamma\vdash \mlstinline{t :: Meas a}$
can be interpreted as a Kleisli morphism, i.e.~a morphism $\Gamma \to \mlstinline{a}$ in $\CatB$. 
If we also have $\mlstinline {x::a},\Delta\vdash \mlstinline{u :: Meas b}$,
then the sequencing $\Gamma,\Delta\vdash \mlstinline{do \{x <- t ; u\} :: Meas b}$ is the following composite morphism:
\[
  \Gamma,\Delta \xrightarrow{\mlstinline{t}\otimes \Delta} \mlstinline{a},\Delta\xrightarrow {\mlstinline{u}} \mlstinline{b}
\]
With this in mind, the commutativity law~(\ref{eqn:comm}) amounts to the interchange law of monoidal categories,
$(\mlstinline{a'}\otimes g)\circ (f\otimes\mlstinline{b})=(f\otimes \mlstinline{b'})\circ (\mlstinline{a}\otimes g)$ (for $f\colon \mlstinline{a} \to \mlstinline{a'}$ and $g\colon \mlstinline{b} \to \mlstinline{b'}$).

We now define two monads by instantiating this lemma. In both cases, $\mathcal V$ is the category of quasi-Borel spaces (\S\ref{sec:monads:qbs}); $\mathcal C$ will be a category of probability (\S\ref{sec:monads:probker}) or measure kernels (\S\ref{sec:monads:measker}). 
\subsection{Key intuition: randomized functions and splitting}
\label{sec:randomizedfun}
Let $\Omega$ be a set of random seeds. A \emph{randomized function}
between sets $X$ and $Y$ is a function
${f\colon X\times \Omega\to Y}$, that depends on the random seed.
Suppose that we have a method for splitting random seeds,
$\gamma:\Omega\to\Omega\times \Omega$ (e.g.~\cite{splitmix}).
Then we can compose randomized functions  
\[f\colon X\times \Omega\to Y\qquad
  g\colon Y\times \Omega\to Z
  \qquad
  (g\circ f)\colon X\times \Omega\to Z
\]
by $(g\circ f)(x,\omega)=g(f(x,\omega_1),\omega_2)$,
where $\gamma(\omega)=(\omega_1,\omega_2)$. 
This is the essence of our treatment of probability. 
However, put plainly like this, composition is not associative. To achieve associativity, we equate certain randomized functions, but to do this we need to talk about expected values, measures and integration (\S\ref{sec:monads:qbs}--\ref{sec:monads:measker}). 

We aim to use Lemma~\ref{lem:freyd-cat} to convert this category into a Kleisli category.
Indeed, by currying, we can regard a function $f\colon X\times \Omega\to Y$ as a
function $X\to Y^\Omega$. Once we equate certain functions, $Y^\Omega$ becomes a monad, and we are thus in the setting of programming with monads. (We emphasise that, despite appearances, this is not the standard reader monad, and $(\Omega,\gamma)$ is not a comonoid.)


\subsection{Rudiments of quasi-Borel spaces}
\label{sec:monads:qbs}

\subsubsection{Rudiments of measure theory}
\label{sec:measure-monads}
We first recall some rudiments of measure theory.
\begin{definition}
  A measurable space $(X,\Sigma_X)$ is a set $X$ together with a set $\Sigma_X$ of `measurable subsets' of $X$, which must be a $\sigma$-algebra, i.e.~closed under countable unions and complements.
  A measure on a space $(X,\Sigma_X)$ is a function $\mu\colon\Sigma_X\to [0,\infty]$ that is $\sigma$-additive
  ($\mu(\biguplus_{i=1}^\infty U_i) = \sum_{i=1}^\infty \mu(U_i)$); it is a probability measure if it is normalized, i.e.~if $\mu(X)=1$. 
  A function $f\colon (X,\Sigma_X)\to (Y,\Sigma_Y)$ is \emph{measurable} if $f^{-1}(U)\in \Sigma_X$ for all $U\in\Sigma_Y$. 
\end{definition}
A key measurable space is $(\RR,\Sigma_\RR)$, where $\Sigma_\RR$ comprises the Borel sets, the least $\sigma$-algebra containing the open intervals. The unit interval $([0,1],\Sigma_{[0,1]})$ is a subspace, and the uniform measure on $[0,1]$ is a measure that assigns to each open interval its length. 
For any measure $\mu$ on $(X,\Sigma_X)$, we can find the expected value of any measurable function $f\colon (X,\Sigma)\to (\RR,\Sigma_\RR)$, notated $\int f(x)\,\mu(\dd x)\in[0,\infty]$, the Lebesgue integral of $f$ with respect to $\mu$. Two measures are the same if they induce the same integration operator. 

\hide{For any measurable space $(X,\Sigma_X)$, the set of probability measures $PX$ can be made into a measurable space,
with $\Sigma_{PX}$ the least $\sigma$-algebra making $\int f(x)\,[-](\dd x):PX \to [0,\infty]$ measurable for all $f\colon X\to \RR$. This is actually a monad, with $\mu\, \mlstinline{>>=} f$ (for $\mu\colon PX, \, f\colon X \to PY$) the measure which is the integration operator taking $g\colon Y\to \RR$ to the iterated integral
\[\int g(y)\,(\mu \, \mlstinline{>>=} \, f)(\dd y) \quad = \quad \int\int g(y) \,f(x,\dd y)\,\mu(\dd x)\text.\]
The category of measurable spaces is not cartesian closed, a point we revisit in Section~\ref{sec:monads}.
But we can actually still interpret the commutativity \eqref{eqn:comm}  and affine laws (\ref{eqn:affine}) for the Giry monad, which respectively amount to Fubini's theorem
\[
  \int\int f(x,y)\,\mu_Y(\dd y)\,\mu_X(\dd x)\ =\ 
 \int\int f(x,y)\,\mu_X(\dd x)\,\mu_Y(\dd y)
\]
and the unity of a probability measure 
\[
  \int 1 \, \mathit{mx}(\dd x) = 1
\]
However, for arbitrary measures, iterated integration does not work, and Fubini does not hold (see \cite{staton-probprog-measure} for a programming angle on this). This is arguably because of obscure measures such as the counting measure which have no role in programming or Monte Carlo simulation. It is an open problem to find a commutative monad of measures on the category of measurable spaces that
includes the one-point measures and the probability measures. We revisit this in Section~\ref{sec:monads:measker}.
}
\subsubsection{Borel spaces and quasi-Borel spaces}


We begin with the notion of \emph{standard Borel space}. In fact, we do not need the traditional definition;
the following characterization will suffice.

\begin{proposition}[e.g.~\cite{kechris}]\begin{enumerate}\item 
    A standard Borel space is a measurable space $(X,\Sigma_X)$ that is either (a)~countable, with $\Sigma_X$ the powerset of $X$, or (b)~measurably isomorphic to $(\RR,\Sigma_\RR)$.
    \item Any measurable subspace of $\RR$ is standard Borel (e.g. $[0,1]$ is standard Borel). 
    \item Standard Borel spaces are closed under countable products.\end{enumerate}\end{proposition}
Let $\Omega$ be a fixed uncountable standard Borel space (traditionally~${\Omega=\RR}$, but see Section~\ref{sec:probspace}).
    \begin{definition}[\cite{qbs}]
      A quasi-Borel space $(X,M_X)$ comprises a set $X$ together with a collection $M_X$ of functions $\Omega\to X$,
      called `admissible random elements',
      such that
      \begin{itemize}
      \item all constant functions are in $M_X$;
      \item composition: if $\alpha\in M_X$ and $f\colon \Omega\to\Omega$ is measurable then $(\alpha\circ f)\in M_X$;
      \item gluing: if $\alpha_1\ldots \alpha_n\ldots \in M_X$ and $\Omega=\biguplus_{n=1}^\infty U_i$ measurable then
        $ \alpha\in M_X$ where $\alpha(\omega)=\alpha_n(\omega)$ when $\omega\in U_i$.
      \end{itemize}
      A morphism $f\colon (X,M_X)\to (Y,M_Y)$ between quasi-Borel spaces is function such that for all $\alpha\in M_X$,
      $(f\circ \alpha) \in M_Y$. 
    \end{definition}
    \begin{proposition}[\cite{qbs}]\label{prop:qbs}
      Quasi-Borel spaces and morphisms form a category $\Qbs$ that is cartesian closed. Standard Borel spaces $(X,\Sigma_X)$ fully embed in $\Qbs$, taking $M_X$ to be the measurable functions. 
    \end{proposition}

\subsection{A category of probability kernels}
\label{sec:monads:probker}
We now revisit the intuition about randomized functions from Section~\ref{sec:randomizedfun} from a more formal perspective.
The key idea is that
$X$ and $Y$ there should be regarded as quasi-Borel spaces and the functions $f,g$ as quasi-Borel functions.
This allows us to equate randomized functions up-to equivalence of measures, giving an affine symmetric monoidal category. 

\subsubsection{Basic probability space}
\label{sec:probspace}
We now fix some basic ingredients:
\begin{itemize}
  \item a standard Borel space $(\Omega,\Sigma_\Omega)$ with a probability measure~$\mu$ on it;
\item a measure-preserving function
\[
  \gamma:(\Omega,\mu)\to(\Omega\times\Omega,\mu\otimes\mu)\text.
\]
i.e.~for all $f\colon \Omega\times \Omega\to \RR$,
$\int f(\gamma(\omega))\,\mu(\dd \omega)=
\int\int f(\omega_1,\omega_2)\,\mu(\dd\omega_2)\,\mu(\dd\omega_1)$;
\item a chosen uniformly distributed random variable $\upsilon\colon \Omega\to [0,1]$.
  \end{itemize}

A canonical example is to let $\Omega = [0,1]^{\NN^*}$, where $\NN^*$ is the set of finite lists of natural numbers, and let
\[\gamma(\omega)=\big(\lambda (i_1, \dots,  i_n). \,\omega(0,i_1, \dots, i_n),\ \lambda (i_1, \dots, i_n). \,\omega(i_1+1, \dots, i_n)\big)\]
In fact, this $\gamma$ is an isomorphism. 
 For an intuition, recall that a list of natural numbers describes a path to a node in the tree that is infinitely deep and infinitely wide (sometimes called a `rose tree').
 So each $\omega\in\Omega$ is an infinitely wide and deep tree where every node is annotated with a real number, and
 $\gamma$ splits the tree as indicated by the dotted line: 
\newcommand{\seq}[1]{\omega(#1)}
\[
  \begin{tikzpicture}
\node (empty) {$\seq{}$};

\node (gammastart) [left=1cm of empty] {};
\node (gammaend) [below=1cm of empty] {};
\node (gamma)  [above left =1mm and -1mm of gammaend,scale=0.7] {$\gamma$};

\draw[dashed,very thick] (gammastart) -- (gammaend);

\node (0) [below left = 5mm and 1cm of empty] {$\seq{0}$};
\node (1) [right=3cm of 0,scale=1] {$\seq{1}$};
\node (2) [right=2 cm of 1,scale=.7] {$\seq{2}$};
\node (3) [right=1cm of 2,scale=.7] {$\seq{3}$};
\node (ldots) [right=0cm of 3,scale=.7] {$\ldots$};

\draw (empty)--(0);
\draw (empty)--(1);
\draw (empty)--(2);
\draw (empty)--(3);

\node (00) [below left = 3mm and 1cm of 0,scale=.9] {$\seq{0,0}$};
\node (01) [right=.8cm of 00,scale=.7] {$\seq{0,1}$};
\node (02) [right=.6 cm of 01,scale=.7] {$\seq{0,2}$};
\node (0ldots) [right=0cm of 02,scale=.6] {$\ldots$};

\draw (0)--(00);
\draw (0)--(01);
\draw (0)--(02);

\node (10) [right = .5cm of 0ldots,scale=.7] {$\seq{1,0}$};
\node (11) [right=.8cm of 10,scale=.7] {$\seq{1,1}$};
\node (12) [right=.6 cm of 11,scale=.7] {$\seq{1,2}$};
\node (1ldots) [right=0mm of 12,scale=.6] {$\ldots$};

\draw (1)--(10);
\draw (1)--(11);
\draw (1)--(12);

\node (vdots2) [below = -2mm of 2,scale=0.6] {$\vdots$};
\node (vdots3) [below = -2mm of 3,scale=0.6] {$\vdots$};
\node (vdots00) [below = -2mm of 00,scale=0.6] {$\vdots$};
\node (vdots01) [below = -2mm of 01,scale=0.6] {$\vdots$};
\node (vdots02) [below = -2mm of 02,scale=0.6] {$\vdots$};
\node (vdots10) [below = -2mm of 10,scale=0.6] {$\vdots$};
\node (vdots11) [below = -2mm of 11,scale=0.6] {$\vdots$};
\node (vdots12) [below = -2mm of 12,scale=0.6] {$\vdots$};
\end{tikzpicture}
\]

Our probability measure~$\mu$ on this choice of $\Omega$
is the countably-infinite product measure of the uniform distribution, given by the Kolmogorov extension theorem. 
For each path $(i_1,\dots, i_n)\in \NN^*$,
the projection function gives a random variable $\Omega\to [0,1]$, which is uniformly distributed, and these are all independent.
In particular, the empty path gives $\upsilon:\Omega\to[0,1]$, with
$\upsilon(\omega)=\omega()$. 

\subsubsection{Probability kernels}
We return to the question of when two randomized functions should be equated (\S\ref{sec:randomizedfun}). Suppose that $\Omega$ is equipped with a probability measure, as in Section~\ref{sec:probspace}.
Although a quasi-Borel space $Y$ is not (a priori) a measurable space,
we still have a construction like a `push-forward measure' along any morphism $f:\Omega\to Y$. By this we mean that for any morphism $h:Y\to \RR$, we can define its `expected value'
to be $\int h(f(\omega))\,\mu(\dd\omega)$. This is a Lebesgue integral because the 
 composite function $hf:\Omega\to \mathbb{R}$ is a morphism, and $\Omega$ and $\mathbb{R}$ are standard Borel spaces, so $hf$ is measurable (Prop.~\ref{prop:qbs}).
In traditional measure theory, two measures inducing the same expectation operation must be equal. We use this intuition to formulate when randomized functions should be equated in this setting. 
\begin{definition}Let $X$ and $Y$ be quasi-Borel spaces.
  A \emph{probability kernel} $f\colon X\rightsquigarrow Y$ is a quasi-Borel function $f\colon X\times \Omega\to Y$. 
  We consider the equivalence relation on probability kernels that is determined by
  \begin{equation}\label{eq:kernel-equiv-relation}
    f\sim g \text{ if for all $x\in X$ and all morphisms $h:Y\to \RR$,} \textstyle\int h(f(x,\omega))\,\mu(\dd \omega)= \int h(g(x,\omega))\, \mu(\dd \omega)\text.
  \end{equation}
\end{definition}
We can perform various constructions on probability kernels:
\begin{itemize}
\item There is a probability kernel $1\rightsquigarrow \RR$
  which describes the uniform distribution on the unit interval $[0,1]$, coming from $\upsilon:\Omega\to [0,1]$. 
\item For any $X$, the identity probability kernel $X\rightsquigarrow X$ is the projection function $X\times \Omega\to X$. 
\item We \emph{compose} two probability kernels $f\colon X\rightsquigarrow Y$, $g\colon Y \rightsquigarrow Z$,
  obtaining a probability kernel $gf\colon X\rightsquigarrow Z$ given by:
  \[
    X\times \Omega\xrightarrow {X\times \gamma} X\times \Omega\times \Omega\xrightarrow {f\times \Omega} Y\times \Omega\xrightarrow g Z
  \]
\item We \emph{tensor} two probability kernels $f\colon A\rightsquigarrow B$, $g\colon X \rightsquigarrow Y$,
  obtaining a probability kernel $f\otimes g\colon (A\times X)\rightsquigarrow (B\times Y)$ given by:
  \[
   A\times X\times \Omega \xrightarrow {A\times X\times \gamma} A\times X\times \Omega\times \Omega \xrightarrow \cong A\times \Omega\times X \times\Omega \xrightarrow {f\times g} B\times Y
 \]
\end{itemize}
\begin{proposition}
  Probability kernels modulo equivalence~(\ref{eq:kernel-equiv-relation}) form a monoidal category~$\ProbKer$: that is, composition and tensor are associative and unital up to equivalence, the interchange law is satisfied up to equivalence, 
  and the operations on probability kernels respect equivalence relations. 
\end{proposition}

We can regard any quasi-Borel function $f\colon X\to Y$ as a probability kernel
\raisebox{0pt}[0pt]{$(X\times \Omega\xrightarrow{\mlstinline{fst}} X\xrightarrow f Y)$}; this induces an identity-on-objects functor $\Qbs\to \ProbKer$. 
\begin{proposition}[see~\cite{qbs}]\label{prop:prob-monad}
  The inclusion functor $\Qbs\to\ProbKer$ has a right adjoint.
  That is, the functions $(\Omega\to X)$ modulo equivalence~(\ref{eq:kernel-equiv-relation}) form an affine commutative monad
  on the category of quasi-Borel spaces. 
\end{proposition}

\subsection{A category of measure kernels}
\label{sec:monads:measker}
 We now turn to unnormalized measures.
The notion of probability kernel on quasi-Borel spaces accounts for the basic notion of pushing forward a probability measure along a
function. 
The other key method for building probability measures, and measures generally, is using \emph{densities} or \emph{weights}.
For example, the density of the beta distribution, $6x(1-x)$, defines a measure on the unit interval. 
Densities can also construct unnormalized measures: starting from the standard normal distribution on $\RR$, the weight
$(\sqrt{2\pi}) e^{\frac 12 x^2}$ defines the Lebesgue measure on $\RR$, which assigns to each open interval its length
(see also~\cite{staton-probprog-measure}).

A parameterized measure, i.e.~a measure kernel, will thus be a probability kernel together with a weight function. 
As motivated in Section~\ref{sec:interface}, this matches the two operations forming measures in the metalanguage,
\lstinline|sample| and \lstinline|score|.
\begin{definition}
  A \emph{measure kernel} $(f,\ell)\colon X\rightsquigarrow Y$ is a pair of quasi-Borel functions, $f:X\times \Omega\to Y$, $\ell\colon X\times \Omega\to [0,\infty]$. 
  We consider the equivalence relation on probability kernels that is determined by
  $(f,\ell)\sim (f',\ell') \text{ if for all $x\in X$ and all morphisms $g:Y\to \RR$,}$\[ \int \ell(x,\omega)\cdot g(f(x,\omega))\ \mu(\dd \omega)= \int  \ell'(x,\omega) \cdot g(f'(x,\omega) )\ \mu(\dd \omega)\text.
  \]
\end{definition}
We can perform various constructions on measure kernels too. 
\begin{itemize}
\item Any probability kernel $X\rightsquigarrow Y$ can be regarded as a measure kernel with constant weight $1$.
\item Any morphism $w\colon X\to \RR$ can be regarded as a measure kernel $X\rightsquigarrow 1$ onto the one point space.
\item We \emph{compose} measure kernels by composing the probability kernels and multiplying the weights. 
\item We \emph{tensor} measure kernels by tensoring the probability kernels and multiplying the weights. 
\item We can regard any quasi-Borel function $X\to Y$ as a measure kernel
$X\rightsquigarrow Y$, with weight constant~$1$; this induces an identity-on-objects functor $\Qbs\to \MeasKer$. 
\end{itemize}
\begin{proposition}[\cite{scibior}, \S4.3.3]\label{prop:meas-monad}
  Measure kernels modulo equivalence form a monoidal category.
  The inclusion functor $\Qbs\to\MeasKer$ has a right adjoint, and 
  so the functions $\Omega\to (X\times \RR)$ modulo equivalence form a commutative monad
  on the category of quasi-Borel spaces. 
\end{proposition}

\subsection{Summary, and alternative approaches and variations}
\subsubsection{Summary}
The category of quasi-Borel spaces provides a model of the metalanguage from Section~\ref{sec:interface}, as follows.
We understand \lstinline|RealNum| as the quasi-Borel space of real numbers, $\RR$. 
\begin{itemize}
\item The \lstinline|Prob| monad is induced by the category of probability kernels, and is affine and commutative, by Proposition~\ref{prop:prob-monad}.
\item The \lstinline|Meas| monad is induced by the category of measure kernels, and is commutative, by Proposition~\ref{prop:meas-monad}.
\item The morphism \lstinline|sample :: Prob a -> Meas a| is induced by the inclusion of the category of probability kernels into the category of measure kernels, taking constant weights ($\ell=1$).
\item The morphism \lstinline|score :: RealNum -> Meas ()| is induced by putting $\ell:\RR\times \Omega\to [0,\infty]$ as the absolute value of the left projection.
\end{itemize}
We also have the concrete distributions used in the examples in Section~\ref{sec:easyexamples}. All distributions on $\RR$ arise as pushforwards of the uniform distribution on $[0,1]$, so they are present;
the morphism \lstinline|iid :: Prob RealNum -> Prob (Stream RealNum)| is defined by the Kolmogorov extension theorem.

As an aside we note another clue that a distinction between \lstinline|Prob| and \lstinline|Meas| is needed. If we had a quasi-Borel morphism \lstinline|iidMeas :: Meas RealNum -> Meas (Stream RealNum)| in this measure-theoretic situation, it could be used to build an infinite-dimensional Lebesgue measure on $\RR^\NN$, which is well-known to be problematic~\cite{lebesgue-R}. By contrast an infinite-dimensional uniform probability distribution is straightforward, by Kolmogorov extension, and is very useful. 

\subsubsection{Categories of measure kernels}
The probability measures on a measurable space $X$ form a affine commutative monad on the category of measurable spaces and measure-preserving maps, called the Giry monad~\cite{giry}.
Moreover, the s-finite measure kernels between measurable spaces form a monoidal category~\cite{staton:sfinite}.
This more traditional measure-theoretic foundation forms a good intuition for our metalanguage (\S\ref{sec:interface}) (see also~\cite{kozen}), but unlike quasi-Borel spaces, it cannot actually model the metalanguage, for two reasons:
first, the category is not cartesian closed, so it does not support all function types~\cite{aumann};
second, it is not known whether the s-finite measure kernels form a Kleisli category.
These issues with the category of measurable spaces led us to use quasi-Borel spaces, where they vanish, and so we use this instead as a semantic basis. 

  \subsubsection{Alternative representations of randomized functions.}
We note a different notion of randomized function,
where the function is equipped with a parameter space (following e.g.~\cite{shiebler}; see also \cite{scibior,lew-trace-types}). 
Let us briefly define a \emph{para-randomized} function
between sets $X$ and $Y$ to be a pair $(\Omega,f)$, where $\Omega$ is a standard probability space and 
$f\colon X\times \Omega\to Y$ is a function. Unlike with our randomized functions, the seed space
is not fixed and is part of the data for a para-randomized function. 
Composition is of the form
\[f\colon X\times \Omega_1\to Y\qquad
  g\colon Y\times \Omega_2\to Z
  \qquad
  (g\circ f)\colon X\times (\Omega_1\times \Omega_2)\to Z
\]
with $(g\circ f)(x,(\omega_1,\omega_2)) =g(f(x,\omega_1),\omega_2)$.
This formulation is convenient when a function has a natural parameter space of fixed dimension such as $\RR^3$. Then composing $(f,\RR^m)$ and $(g,\RR^n)$
yields $(g\circ f,\RR^{m+n})$. This is reasonable for certain simple probabilistic programs, but in this article we are especially interested in the situation where the parameter spaces are not so simple. For example, in \S\ref{sec:easyexamples:regpiecewiselin} we compose each point of a Poisson point process, of infinite dimension,
with a random linear function; it is not so clear how to manage this straightforwardly by combining dimensions.


\subsubsection{Domain theoretic models and recursion}
The equational definitions in Sections~\ref{sec:interface}--\ref{sec:randfun} all have solutions in quasi-Borel spaces.
For general higher order recursion, one can extend quasi-Borel spaces by placing a cpo structure on the carrier, and imposing compatibility conditions, see~\cite{wqbs}. We omit the details, which are largely orthogonal. 

A different approach would be to use the recent work of~\cite{jgl-domain-probprog} to build an entirely domain theoretic model. 
There are also other potential ways to interpret laziness in probabilistic models of linear logic \cite{measurable-cones,geometryDLH,maraist1999call}.
 Here we have stuck with quasi-Borel spaces because they also connect to our implementation~(\S\ref{sec:monads:impl}).

\subsubsection{Open questions}
We briefly remark that although the category of quasi-Borel spaces accommodates all the examples from Section~\ref{sec:easyexamples}, the abstraction and generality of the metalanguage opens up challenges. For example, we do have \lstinline|iid| for spaces that are standard Borel, but Kolmogorov extension beyond that is an open question; semantic models of stochastic memoization for uncountable domain spaces (e.g.~white noise) is an open question; semantic models with Gaussian processes with discontinuous kernels is an open question.
These are all well-known challenges for traditional measure-theoretic probability, but they can now be phrased in precise terms through the metalanguage. 
(Curiously, none of these things are at all problematic in our implementation (\S\ref{sec:monads:impl}).)

\subsection{Haskell implementation of the probability and measures monads}
\label{sec:monads:impl}
In Sections~\ref{sec:randomizedfun}--\ref{sec:monads:measker} we gave a semantic model of the metalanguage (\S\ref{sec:interface}). This leads directly to our implementation in the LazyPPL library~\cite{nonanongit}.
%
For the probability space, we put $\Omega=\mlstinline{Tree}$ and $\gamma=\mlstinline{splitTree}$:
\hide{\begin{lstlisting}[columns=flexible]
 data Tree = Tree Double [Tree]
 splitTree :: Tree -> (Tree , Tree)
 splitTree (Tree r (t : ts)) = (t , Tree r ts)
\end{lstlisting}}
\[\begin{array}{@{}p{6cm}|p{8cm}}\mlstinlin{data Tree = Tree Double [Tree]}
 &\mlstinlin{splitTree :: Tree -> (Tree , Tree)}
 \\&\mlstinlin{splitTree (Tree r (t : ts)) = (t , Tree r ts)}
\end{array}
\]
A probability distribution over \lstinline|a| is a function \lstinline|Tree -> a|.
\hide{\begin{lstlisting}
 newtype Prob a = Prob (Tree -> a)
 uniform :: Prob Double
 uniform = Prob $ \(Tree r _) -> r 
 return a = Prob $ const a
 (Prob m) >>= f = Prob $ \g -> let (g(*@$\cb_1$@*),g(*@$\cb_2$@*)) = splitTree g
                               (Prob m') = f (m g(*@$\cb_1$@*))
                               in m' g(*@$\cb_2$@*)
\end{lstlisting}}
\hide{\begin{lstlisting}[columns=flexible]
 newtype Prob a = Prob (Tree -> a)
 uniform :: Prob Double
 uniform = Prob $ \(Tree r _) -> r 
 return a = Prob $ const a
 (Prob m) >>= f = Prob $ \g -> let {(g(*@$\cb_1$@*),g(*@$\cb_2$@*)) = splitTree g; (Prob m') = f (m g(*@$\cb_1$@*))} in m' g(*@$\cb_2$@*)
\end{lstlisting}}
\begin{align*}
&\begin{array}{@{}p{6cm}|p{8cm}}\mlstinlin{newtype Prob a = Prob (Tree -> a)}\qquad
&\mlstinlin{uniform :: Prob Double}
\\ \mlstinlin{return a = Prob $ const a}
   &\mlstinlin{uniform = Prob $ \(Tree r _) -> r }
     \end{array}\\&
 \mlstinlin{(Prob m) >>= f = Prob $ \\g -> let \{}\mlstinline{(g$\cb_1$,g$\cb_2$) = splitTree g; (Prob m') = f (m g$\cb_1$)}\mlstinline{\} in m' g$\cb_2$}
\end{align*}
Note that although the type looks like the reader monad, the bind is different. 
A similar bind is used in QuickCheck~\cite[\S6.4]{quickcheck}, although we are not aware of 
a semantic analysis in the literature. 

We implement the measures monad using the writer monad transformer. 
Because weights multiply, they often become very small, and so we use log numbers. 
\newline

\noindent\hide{\begin{lstlisting}[columns=flexible]
 newtype Meas a = Meas (WriterT (Product (Log Double)) Prob a)
 score :: Double -> Meas ()
 score r = Meas $ tell $ Product $ (Exp . log) $ r
 sample :: Prob a -> Meas a
 sample p = Meas $ lift p
\end{lstlisting}}%
$\begin{aligned}
&\mlstinline{newtype Meas a = Meas (WriterT (Product (Log Double)) Prob a)}\\
 &\begin{array}{@{}p{6cm}|p{8cm}}
                                             \mlstinlin{sample :: Prob a -> Meas a}&
    \mlstinlin{score :: Double -> Meas ()}
\\\mlstinlin{sample p = Meas \$ lift p}\quad
& \mlstinlin{score r = Meas \$ tell \$ Product \$ (Exp . log) r}\end{array}
\end{aligned}$
\newline

\label{sec:monads:lwis}
We provide several inference methods in~\cite{nonanongit}. A simple reference method is a likelihood-weighted importance sampler, following e.g.~\cite[\S4.1]{meent-book}, which is just 20 lines of Haskell. 
This has type
\lstinline|lwis :: Int -> Meas a -> IO [a]|.
Running~\lstinline|(lwis |$n$\lstinline| m)| produces a stream of samples from the unnormalized measure \lstinline|m|.
\hide{and uses the following pseudocode: 
\begin{enumerate}
\item generate $n$ pairs of weighted samples from \lstinline|m|, i.e.~pairs
$(w_i,x_i)$ of a weight $w_i$ and result $x_i$ in~\lstinline|a|, each according to the following algorithm:
\begin{enumerate}
\item lazily initialize a random rose tree \lstinline|t :: Tree|, allocating a fresh random number to each~node;  
\item run the program \lstinline|m :: Meas a| with the rose tree \lstinline|t|; 
\item return the pair $(w,x)$ of the resulting value and the accumulated weight. 
\end{enumerate}
\item generate an infinite stream of samples from the discrete `empirical' distribution 
$\sum_{i=1}^n \frac{w_i}{\sum_{i=1}^n w_i} \cdot x_i$:
\begin{enumerate}
\item at each step, pick a random number $r$ uniformly in the interval $[0,\sum_{i=1}^nw_i]$;
\item if $r$ is in the interval $[\sum_{i=1}^{j-1} w_i,\sum_{i=1}^j w_j]$ then output $x_j$.
\end{enumerate}
\end{enumerate}
As $n\to \infty$, the stream of samples appears to be  almost surely converges to the normalized probability distribution. 
But in practice, for small $n$, it is not usually very good, and so we look at a better algorithm in Section~\ref{sec:msmh:mhg}. }
As $n\to \infty$, the stream of samples converges to a stream of iid samples from the normalized probability distribution corresponding to \lstinline|m|. 
But in practice, for feasible $n$, it is not usually very accurate, and so we look at a better algorithm in Section~\ref{sec:msmh}.


%% file: msmh.tex

\label{sec:msmh}

In Section~\ref{sec:monads} we showed that a closed program in the metalanguage (\S\ref{sec:interface}) of
type \mlstinline{Meas $\ X$} induces a pair of functions
\raisebox{0pt}[11pt]{$\RR\xleftarrow \ell \Omega \xrightarrow {f} X$},
where $\Omega$ is regarded with a basic probability measure~$p$,
and $\ell$ is measurable. Here $\ell$ is regarded as a density for an unnormalized distribution on $\Omega$,
which is then to be pushed forward to~$X$, which is the space of interest. 
We now describe a new Markov-Chain Monte Carlo inference algorithm that works for programs in the metalanguage under this semantics, which works well for the examples we have considered here (\S\ref{sec:easyexamples},\ref{sec:randfun}).

Notice that there are four measures of interest:
\begin{itemize}
\item The basic probability measure $\mu$ on $\Omega$;
\item The unnormalized measure $\mu_\ell$ on $\Omega$, induced by regarding $\ell$ as a density. Formally,
  $\mu_\ell(U)=\int_\Omega [\omega\in U]\cdot \ell(\omega)\,\mu(\dd \omega)$.
  This could be written in the metalanguage as \[\mlstinline{do \{$\cb \omega\ $ <- sample $\cb\ \mu$ ; score ($\cb \ell\,\omega$); return $\cb\ \omega$\} :: Meas $\cb\ \Omega$}
    \text.\]
  (Here, and throughout Section~\ref{sec:msmh:green} and the proof of Thm.~\ref{thm:mutatetrees-green}, we are using the metalanguage to discuss and manipulate semantic measures -- these are not necessarily programs to be run directly, by contrast with Section~\ref{sec:easyexamples}.)
\item The \emph{normalized} form of the measure $\mu_\ell$, $\frac {\mu_\ell} {\mu_\ell(\Omega)}$, which is a probability measure, assuming $\mu_\ell(\Omega)\in (0,\infty)$. 
\item The pushforward probability measure on $X$, $f^*(\frac {\mu_\ell} {\mu_\ell(\Omega)})$.
  This could be written \[\mlstinline{do \{$\cb \omega\ $ <- sample $\cb\ \mu$ ; score ($\cb \ell\,\omega$); score ($1/({\cb \mu_\ell}\ {\cb\Omega})$) ; return ($\cb f\,\omega$)\} :: Meas $\cb\ X$}\text.\]
\end{itemize}
The challenge is that the normalizing constant $\mu_\ell(\Omega)$ is typically difficult to calculate.
The Markov-Chain Monte Carlo simulation algorithms provide a sampling procedure for $\frac {\mu_\ell} {\mu_\ell(\Omega)}$ on $\Omega$, without explicitly calculating $\mu_\ell(\Omega)$.
They are best described as algorithms over $\Omega$, rather than $X$, although we can push-forward the samples to $X$ at the last minute.

\subsection{Proposal kernels in general}
\label{sec:msmh:green}

The key ingredient for a Metropolis-Hastings algorithm is a `proposal' Markov kernel. 
This is a function $k:\Omega\times \Sigma_\Omega\to [0,1]$ such that each $k(\omega,-)$ is a probability measure and each $k(-,U)$ is measurable. We follow the analysis of proposal kernels from~\cite{rjmcmc,geyer}.

The proposal kernel~$k$ does not directly capture the probability measure  $\frac {\mu_\ell} {\mu_\ell(\Omega)}$. Rather, it induces another kernel, which works by first proposing changes (using $k$) and then either accepting or rejecting them
(\S\ref{sec:msmh:mhg}). This depends on an `acceptance ratio' which exists as long as $k$ satisfies the Green property, that we define now. 

Given a Markov kernel $k$ we can form an unnormalized kernel by composing it with the unnormalized measure $\mu_\ell$.
This gives two measures $m,m_{\mathit{rev}}$ on $\Omega\times \Omega$:
\begin{align*}
  m(U)\ &=\ \int_\Omega\int_\Omega [(\omega_1,\omega_2)\in U] \cdot \ell(\omega_1)\, k(\omega_1,\dd \omega_2)\, \mu(\dd \omega_1)\\
  m_{\mathit{rev}}(U)\ &=\ \int_\Omega\int_\Omega [(\omega_2,\omega_1)\in U] \cdot \ell(\omega_1)\, k(\omega_1,\dd \omega_2)\, \mu(\dd \omega_1)
\end{align*}
These can be described in programming terms as
\begin{align*}m&=\mlstinline{do \{$\cb\omega_1\ $ <- sample $\;\cb\mu\ $ ; $\ \cb\omega_2\ $ <- sample ($\cb k\,\omega_1$) ; score ($\cb l \omega_1$) ; return ($\cb \omega_1$,$\cb\omega_2$)\}}
  \\  m_{\mathit{rev}}\ &= \mlstinline{do \{$\cb\omega_1\ $ <- sample $\;\cb \mu\ $ ; $\ \cb\omega_2\ $ <- sample ($\cb k\,\omega_1$) ; score ($\cb l\,\omega_1$) ; return ($\cb\omega_2$,$\cb\omega_1$)\}}
\end{align*}
\begin{definition}[\cite{rjmcmc}]\label{def:greenk}
  We say that a kernel $k$ is \emph{Green} with respect to $\ell$ and $\mu$ if
  $m_{\mathit{rev}} $ is absolutely continuous with respect to $m$. This means that there exists a `ratio' 
  $r:\Omega\times \Omega\to \RR$ (the `Radon-Nikodym derivative') such that
  \[\int [(\omega_1,\omega_2)\in U] \cdot r(\omega_1,\omega_2)\cdot \ell(\omega_1)\ k(\omega_1,\dd \omega_2)\,\mu(\dd\omega_1)
    =\int [(\omega_2,\omega_1)\in U] \cdot \ell(\omega_1)\ k(\omega_1,\dd \omega_2)\,\mu(\dd\omega_1)
  \]
  or in programming terms
  \begin{align*}
    &\mlstinline{do \{$\cb\omega_1\ $ <- sample $\;\cb \mu\ $; $\ \cb\omega_2\ $ <- sample ($\cb k\,\omega_1$); score ($\cb \ell\,\omega_1$); score ($\cb r \;\omega_1\;$ $\cb\omega_2$); return ($\cb\omega_1$,$\cb\omega_2$)\}}\\&=
          \mlstinline{do \{$\cb\omega_1\ $ <- sample $\;\cb \mu\ $; $\ \cb\omega_2\ $ <- sample ($\cb k\, \omega_1$); score ($\cb \ell\, \omega_1$); return ($\cb\omega_2$,$\cb\omega_1$)\}}
         \end{align*}
\end{definition}

\subsection{A new proposal kernel for lazy rose trees}
\label{sec:msmh:rose}
Recall our choice of $\Omega$ is rose trees: infinitely deep and infinitely wide trees labelled from $[0,1]$, with the basic probability measure~$\mu$ giving the
uniform distribution to all nodes.
We consider a new proposal kernel, parameterized by a probability $p\in[0,1]$:
\begin{itemize}\item for every node, toss a coin with bias~$p$; if heads, resample from the uniform distribution on $[0,1]$, if tails, leave it alone.\end{itemize}
This requires an infinite number of changes, but since probability is treated lazily, there is no problem in practice. 
\begin{lstlisting}
 mutateTree :: RealNum -> Tree -> Prob Tree
 mutateTree p (Tree a ts) = do b <- bernoulli p
                               a' <- uniform
                               ts' <- mapM (mutateTree p) ts
                               return $ Tree (if b then a' else a) ts'
   \end{lstlisting}
   This can be defined measure-theoretically using Kolmogorov's extension theorem. 
\begin{theorem}\label{thm:mutatetrees-green}
  The kernel $k:\Omega\times \Sigma_\Omega\to [0,1]$ given by \lstinline|(mutateTree p)| is Green, 
  and the ratio is $r(\omega_1,\omega_2)=\frac{\ell(\omega_2)}{\ell(\omega_1)}$. 
\end{theorem}
\begin{proof}[Proof notes]
  Notice that $k$ is reversible with respect to $\mu$ in that
  \[\mlstinline{do \{$\cb\omega_1\ $ <- $\ \cb\mu$; $\ \cb\omega_2\ $ <- $\ \cb k\,\omega_1$; return ($\cb\omega_1$,$\cb\omega_2$)\}} \ \ = \ \ \mlstinline{do \{$\cb\omega_1\ $ <- $\ \cb \mu$; $\ \cb\omega_2\ $ <- $\ \cb k\,\omega_1$; return ($\cb\omega_2$,$\cb\omega_1$)\}}\]
  This can be deduced from Kolmogorov's extension theorem, by proving it for finite projections. 
  Therefore the given $r$ is indeed a ratio, since
  \begin{align*}
    &\mlstinline{do \{$\cb\omega_1\;$ <- sample $\;\cb \mu$; $\;\cb\omega_2\;$ <- sample ($\cb k \,\omega_1$); score ($\cb\ell\,\omega_1$); score ($\cb r\,\omega_1\,\omega_2$); return ($\cb\omega_1$,$\cb\omega_2$)\}}
    \\&=
  \mlstinline{do \{$\cb\omega_1\;$ <- sample $\;\cb \mu$; $\;\cb\omega_2\;$ <- sample ($\cb k\,\omega_1$); score ($\cb\ell\,\omega_2$); return ($\cb\omega_1$,$\cb\omega_2$)\}}
    \\&=
  \mlstinline{do \{$\cb\omega_1\;$ <- sample $\;\cb \mu$; $\;\cb\omega_2\;$ <- sample ($\cb k \,\omega_1$); score ($\cb\ell\,\omega_1$); return ($\cb\omega_2$,$\cb\omega_1$)\}}
    \end{align*}
as required, where the second step uses the reversibility of $k$ with respect to $\mu$. 
\end{proof}
\paragraph{Technical note.} Our $k$ is reversible in the given sense, and this appears to be a `Metropolis ratio'.
But because our space $\Omega$ is infinite-dimensional, the traditional density-based analysis 
of Metropolis does not apply, whereas this more general approach by Green does. 
\subsection{The Metropolis-Hastings-Green Markov Chain}
\label{sec:msmh:mhg}
 Let $k: \Omega\times \Sigma_\Omega\to [0,1]$ be a Green Markov kernel with ratio $r:\Omega\times\Omega\to [0,1]$.  
The \emph{Metropolis-Hastings-Green kernel} $k_{\mathit{MHG}}: \Omega\times \Sigma_\Omega\to [0,1]$ is now given by
\emph{proposing} a new $\omega_2$ via $k(\omega_1,-)$, and accepting or rejecting the proposal according to $\min(1,r(\omega_1,\omega_2))$. 
Either way, we produce something, either the new $\omega_2$ or the old $\omega_1$. 
\begin{lstlisting}[columns=flexible]
 (*@$\cb k_{\mathit{MHG}}$@*) :: (*@$\Omega $@*) -> Prob (*@$\Omega $@*)
 (*@$\cb k_{\mathit{MHG}}$@*) (*@$\cb\omega_1$@*) = 
  do {(*@$\cb\omega_2$@*) <- (*@$\cb k\,\omega_1$@*) ; b <- bernoulli $ min 1 ((*@$\cb r\,\omega_1\,\omega_2$@*)) ; if b then return (*@$\cb\omega_2$@*) else return (*@$\cb\omega_1$@*)}
\end{lstlisting}
We can then construct a Markov chain with transitions given by $k_{\mathit{MHG}}$. The key result (e.g. \cite{rjmcmc,geyer}) is that when $k$ is well-behaved, the states of this Markov chain approximate the posterior distribution. Theorem~\ref{th:mhg} says this formally. Recall that a probability measure $\nu$ on $\Omega$ is a stationary distribution for a kernel $k : \Omega \times \Sigma_\Omega \to [0, 1]$ if 
\[
\int_{\Omega} k(\omega, U) \cdot  \nu(\dd \omega) = \nu(U).
\]
We say that $k$ is irreducible with respect to a probability measure $\xi$ if for every $\omega \in \Omega$ and for every $U \in \Sigma_\Omega$ such that $\xi(U) > 0$, there exists $n \in \NN$ such that $k^n(\omega, U) > 0$. Informally, irreducibility means that the Markov chain will reach any set of positive measure in finite time.
\begin{theorem}[Metropolis-Hastings-Green]\label{th:mhg}
For any Green kernel $k$, the induced kernel $k_{\mathit{MHG}}$ has a stationary distribution, which is the normalized probability measure 
$\frac{\mu_\ell}{\mu_\ell(\Omega)}$ on $\Omega$. 
If $k_{\mathit{MHG}}$ is irreducible with respect to $\frac{\mu_\ell}{\mu_\ell(\Omega)}$ then the stationary distribution is unique.
\end{theorem}
We can therefore use the Metropolis-Hastings-Green kernel as a method for sampling from the normalized probability measure.
\begin{proposition}\label{prop:mh1irr}
  For the \lstinline|mutateTree| kernel (\S\ref{sec:msmh:rose}) with $p=1$, $k_{\mathit{MHG}}$ is irreducible for $\frac{\mu_\ell}{\mu_\ell(\Omega)}$.
\end{proposition}
\begin{proof}[Proof note]
  Here $n=1$ suffices. 
\end{proof}
We recall that correctness of a similar `all-sites' Metropolis-Hastings scheme for probabilistic programming was proved in~\cite{bdgs-lambdacalc-prob},
albeit for a non-lazy language.

There remains a concern that $k_{\mathit{MHG}}$ is not irreducible for $p<1$.
Indeed, in that situation, the set $U=\{\omega'~|~\forall i.\,\omega_i\neq \omega_i'\}$ is not reachable from $\omega$, even though $U$ typically has measure~$1$.  More informally, although every node has a chance of being changed, there will almost surely exist a node that is not changed.
In practice, \lstinline|(mutateTree p)| alone appears to be fine, because any finite collection of samples will only invoke a finite number of nodes anyway. We return to this point in Section~\ref{sec:ssmhmix}.
\subsection{Summary and example}\label{sec:msmh:summary}
In summary, we have a procedure for sampling from the distribution described by a program in the metalanguage (\S\ref{sec:interface}), by using the Metropolis-Hastings-Green kernel (\S\ref{sec:msmh:mhg})
associated to the Green Markov kernel~\lstinline|(mutateTree p)| (\S\ref{sec:msmh:rose}). 
Each step of the algorithm provides a sample from the measure $\frac{\mu_\ell}{\mu_\ell(\Omega)}$ on $\Omega$,
and we can push-forward this sample along $f\colon \Omega\to X$ to obtain a sample from the measure described by the program.

To illustrate, we recall the simple linear regression model (\S\ref{sec:easyexamples:reglin}). Although we are using an infinite tree, only two samples will be used, for the slope~\lstinline|a| and 
intercept~\lstinline|b|. 
If we use our kernel with $p=0.5$, at each step, \emph{one} of the following steps will happen, each with probability $0.25$. 
\begin{itemize}
\item We will change neither \lstinline|a| nor~\lstinline|b|. (This is a wasted step.)
\item We try to change the slope \lstinline|a| but keep the intercept \lstinline|b| the same. This is useful if they are independent. 
\item We try to change the intercept \lstinline|b| but keep the slope \lstinline|a| the same. Again, this is useful if they are independent. 
\item We try to change both the slope \lstinline|a| and the intercept \lstinline|b|. This is sometimes called `multisite' inference, and is useful if they are correlated.
\end{itemize}
As is always the case with general purpose methods, it is non-optimal if the independence and correlations are known. But our algorithm serves well where they are not known, and moreover works perfectly well with the lazy structures used in the probability monad.


%% file: singlesite.tex

\label{sec:single-site}



There are many possible variations on the generic inference algorithm in Section~\ref{sec:msmh}.
In this section we consider the mixture of kernels --- randomly choosing between different kernels at each step. 
This has at least two useful applications:
\begin{itemize}
\item Mixing \lstinline|(mutateTree p)| with \lstinline|(mutateTree 1)|, which ensures a unique stationary distribution, and intuitively allows the entire tree to be reset sometimes, which can be useful for exploring multimodal distributions (\S\ref{sec:ssmhmix});
\item A single-site proposal kernel, where we mutate exactly one node in each step, following~Wingate et al.~\cite{lightweight-mh} (\S\ref{sec:single-site:proposal}).
\end{itemize}
\hide{The inference algorithm in Section~\ref{sec:msmh} randomly mutates each
node (also sometimes referred to as \emph{site}) in the lazy rose tree
with some given bias $p$ at every proposal step.
Such an inference technique may not always be desirable, however,
e.g.~in situations where we want to control the exact number of nodes being
mutated in each proposal (so that our mutations are more gradual): we can never guarantee
a constant fixed number of mutations due to them happening independently of each other. }

\hide{suggest a proposal kernel where
only a single node is mutated at every proposal step. In this Section we discuss our own implementation of this
in LazyPPL. We found this to be particularly profitable where it is unclear which~$p$ to use, as in the clustering example (\S\ref{sec:dp}).}

\hide{\subsection{Overview}
Recall that our source of randomness comes from an infinitely deep and wide
tree. Here we suggest a proposal kernel \lstinline|mutateSite| roughly of the type
\lstinline|Tree -> Prob Tree|, which works by only mutating a \emph{single} uniformly
random node (aka site) in our tree. We say ``roughly'' because we do not uniformly 
randomly choose a site from a countably infinite set (which is impossible).
Instead, we consider the \emph{finite} set of sites in the tree accessed by the
probabilistic program and mutate a random site from there.
Intuitively the idea
is to only change sites which directly affect the current computation, ensuring fewer
wasteful runs of the proposal step where irrelevant sites are changed.
This is very stateful...
}
\subsection{State-dependent mixing in general}
\label{sec:single-site:theory}
 We consider the following general method for mixing Green kernels (Def.~\ref{def:greenk}), which is perhaps implicit in~\cite{rjmcmc,geyer}. 
Let $k_i:\Omega\times \Sigma_\Omega\to[0,1]$ be a countable family of Markov kernels (\S\ref{sec:msmh:green}),
and let $c\colon \Omega\times\NN\to [0,1]$ be a parameterized probability distribution function over $\NN$,
i.e.~for all $\omega\in\Omega$, $\sum_{i=1}^\infty c(\omega,i)=1$.
Let $k$ be the mixed kernel
\begin{equation}k(\omega,U)=\textstyle \sum_{i=1}^\infty c(\omega,i)\cdot k_i(\omega,U)\text.\label{eqn:mixkernel}\end{equation}
  Suppose that each $k_i$ is a Green kernel with respect to $\mu$, with 
  ratio $r_i:\Omega\times\Omega\to[0,\infty]$ (Def.~\ref{def:greenk}). 
Suppose that we can always detect which kernel was used, i.e.~there is a function
$e:\Omega\times\Omega\to \NN$ such that
\mlstinline{do \{$\omega$ <- $\mu$; $\ {\omega^\prime}$ <- $k_i\,\omega$; return ($\omega$,$\omega^\prime$,$e(\omega,\omega^\prime)$)\} = do \{$\omega$ <- $\mu$; $\ \omega^\prime$ <- $k_i\,\omega$; return ($\omega$,$\omega^\prime$,$i$)\}}. 

\begin{theorem}\label{thm:state-dependent-mixing}
  The kernel $k$ \eqref{eqn:mixkernel} is a Green kernel with respect to $\mu$, with ratio
  $r\colon \Omega\times \Omega\to [0,\infty]$
  \[r(\omega,\omega')=r_i(\omega,\omega')\cdot \frac{c(\omega',i)}{c(\omega,i)}
    \quad\text{where $i=e(\omega,\omega')$\text.}
  \]
\end{theorem}
\subsection{Application 1: mixing for unique stationary distributions}\label{sec:ssmhmix}
A special case of Theorem~\ref{thm:state-dependent-mixing} is independent mixing, which is a simple way of combining kernels, and a way of building irreducible MHG kernels (i.e.~with unique stationary distributions, \S\ref{sec:msmh:mhg}):
\begin{proposition}
  If $k$ and $k'$ are Green kernels, and $k$ induces an irreducible MHG kernel, then for $r\in (0,1)$,  the kernel $(r\cdot k + (1-r)\cdot k')$ again induces an irreducible MHG kernel. 
  \end{proposition}
  \begin{proof}[Proof notes]
    The kernel $(r\cdot k + (1-r)\cdot k')$ is Green by Theorem~\ref{thm:state-dependent-mixing}. 
    For irreducibility, consider~$\omega$ and $U$; since $k$ is irreducible there is $n$ such that $k^n(\omega,U)>0$, so $(r\cdot k + (1-r)\cdot k')^n(\omega,U)> r^n\cdot k^n(\omega,U) > 0$. 
  \end{proof}
  In particular, completing the discussion from Section~\ref{sec:msmh:mhg}, \lstinline|(mutateTree p)| might not induce an irreducible kernel, but if an irreducible kernel is desired then we can instead start from the mixed kernel
  $(r\cdot (\mlstinline{mutateTree 1}) + (1-r)\cdot (\mlstinline{mutateTree p}))$, using Prop.~\ref{prop:mh1irr}.

\subsection{Application 2: single-site proposal kernel for lazy rose trees}
\label{sec:single-site:proposal}
The single-site Metropolis-Hastings proposal kernel is a popular generic kernel for probabilistic programming (\cite{lightweight-mh},\cite[\S4.2.1]{meent-book}). In a finite-dimensional situation, the idea is to randomly pick one dimension and change it, leaving the other dimensions unchanged.

The subtlety here is that in a lazy program, such as the examples in Sections~\ref{sec:easyexamples} and~\ref{sec:randfun}, the number of dimensions is unbounded, and so it is a priori impossible to pick one dimension uniformly at random. Nonetheless, we now use Theorem~\ref{thm:state-dependent-mixing} to show that there is actually a well-behaved analogue of this kernel that is relevant where there is lazy structure. The idea is to inspect the dimensions that are actually used in a given computation.

\paragraph{High-level view}
\label{sec:single-site-high-level}
Recall the representation of probabilistic programs developed in Sections~\ref{sec:monads} and~\ref{sec:msmh}, with $\Omega$ the infinite rose trees,
and weight function $\ell:\Omega\to[0,\infty]$, and an outcome $\Omega\to X$.
We describe the single-site proposal kernel at this level.
To do this, we instantiate state-dependent mixing as follows. We work up-to a bijection between natural numbers and finite paths through the rose tree, which are countably infinite. 
\begin{itemize}
\item For each path~$i$ through the rose tree, let $k_i$ be the kernel that randomly
  changes node~$i$ and leaves the others unchanged.
  This is a Green kernel with ratio $\ell(\omega')/\ell(\omega)$.
\item If $\omega$ and $\omega'$ differ by only one node, then let $e(\omega,\omega')$ return the path to this node.
\item For any given tree $\omega$, we define $c(\omega,i)$ as follows.
  First, if $\ell$ and $f$ are defined by programs that terminate, then we calculate the necessarily finite set of nodes $S_\omega=\{i_1\dots i_n\}$ that are actually inspected in evaluating $\ell(\omega)$ and $f(\omega)$.
  We then pick one at random, i.e.~let $c(\omega,i)= \frac 1{|S_\omega|}$ if $i\in S_\omega$, and $c(\omega,i)= 0$ otherwise.
\item Following Theorem~\ref{thm:state-dependent-mixing}, we can calculate the Green ratio as $\frac{\ell(\omega')\cdot |S_{\omega}|}{\ell(\omega)\cdot |S_{\omega'}|}$. 
\end{itemize}

\paragraph{Implementation details}
Lazy evaluation is the sole reason why we are even able to consider `the set of nodes
in the tree that have been evaluated' in any given run of our probabilistic program, and
it ensures that no irrelevant sites are present in that set (i.e. those sites which do
not affect the outcome of the result of that run). In our implementation of the single-side proposal kernel, we go under the hood and inspect system
memory from within Haskell to calculate this set of sites $S_\omega$.
The \lstinline|ghc-heap| module
exposes the parts of the tree that have been evaluated 
to weak-head normal form and parts still
untouched (present in memory as thunks), and using this we can safely inspect the runtime
evaluation state of our random tree without forcing any further computation on it.

\paragraph{A note on performance}
General purpose Metropolis-Hastings methods cannot be expected to out-perform hand-crafted methods, or more specialized Monte Carlo algorithms \cite{meent-book}, and so we have not carried out a detailed performance evaluation, but the performance was perfectly adequate for the examples in Sections~\ref{sec:easyexamples} and~\ref{sec:randfun}. Our main intention for this section was to demonstrate that the popular single-site proposal kernels can still be used in the lazy setting.
That said, we can make some brief anecdotal remarks: we found that sometimes, where there are very many independent sites, such as in cluster assignment, the single-site method will perform better, whereas in many other situations the method of Section~\ref{sec:msmh} performs well; in multimodal situations, such as mixture models, we found it is beneficial to use Section~\ref{sec:ssmhmix}, allowing a small chance of resetting the tree and exploring a totally different region. 

Our illustrations in Figures~\ref{fig:regression} and~\ref{fig:ibp} (see~\cite{nonanongit} for more) show that these inference methods produce reasonable results in practice. We also remark briefly on the time/space usage of the current implementation in Haskell. The three regression illustrations in Figure~\ref{fig:regression} took around 2s, 18s and 19s respectively (AWS EC2 t3.large). For instance, in Figure~\ref{fig:regression}(b) we ran $10^6$ steps of the Markov chain, thinning the samples for legibility. Although these timings are fine, in our experience, memory usage can be a concern with larger models in the current implementation. For example, we tested a program induction example which is highly multimodal (see~\cite{nonanongit}), and it typically needs more than 8GB RAM to produce $10^6$ samples. 



\hide{\subsubsection{Calculating the nodes that are used in evaluation}
Lazy evaluation is the sole reason why we are even able to consider `the set of nodes
in the tree that have been evaluated' in any given run of our probabilistic program, and
it ensures that no irrelevant sites are present in that set (i.e. those sites which do
not affect the outcome of the result of that run). By going under the hood and inspecting system
memory (from within Haskell) we are able to calculate this set of sites $S_\omega$.

Crucially, we make use of GHC's \lstinline|ghc-heap| module, which exposes the layout
of heap objects in memory allowing us to view function closures, pointers to thunked
objects, and several other GHC runtime-specific objects present in memory.  By analysing
the parts of the tree that have been evaluated to weak-head normal form and parts still
untouched (present in memory as thunks), we are able to safely inspect the \emph{runtime
evaluation state} of our random tree without forcing any further computation on it.
We encapsulate our use of \lstinline|ghc-heap| in a method truncating lazy infinite rose trees into finite partial trees, as follows. 
We first define the type of partial trees where unevaluated branches are explicitly marked.
\begin{lstlisting}
 data PTree = PTree (Maybe Double) [Maybe PTree]
\end{lstlisting}
Following this we define \lstinline|trunc :: Tree -> IO PTree|, the key function we use in
our implementation of single-site MH. It takes a lazy rose tree as input, inspects its structure in
the Haskell heap, and from it builds the partial tree.
For example, if the evaluation of a program forces the partial evaluation of a rose tree in memory
to be \lstinline[breaklines=true]|t = Tree _ (_ : (Tree 0.3 _) : _ : (Tree 0.4 _) : _)|
(underscores represent unevaluated thunks), its \lstinline|trunc|ation will return
\lstinline[breaklines=true,keepspaces=true]|PTree Nothing [Nothing, Just (PTree (Just 0.3) []), Nothing, Just (PTree (Just 0.4) [])]|.
It is from this partial tree that we identify the finite set of used nodes $S_\omega$, which is the key ingredient in the proposal step (\S\ref{sec:single-site-high-level}). We then uniformly pick a site to modify, and follow the generic Metropolis-Hastings-Green algorithm (\S\ref{sec:msmh:mhg}).}

\hide{\subsubsection{A note on irreducibility}
\label{sec:single-site:irreducibility}
A Markov chain constructed using our single-site MH need not always be irreducible: this depends on the program.
This can immediately be seen in the coin-toss example \lstinline|test| from \S\ref{sec:hardexamples:addressing}
where the only way to go from \lstinline|(x,y) = (True,True)| to \lstinline|(False,False)| is via either 
\lstinline|(True,False)| or \lstinline|(False,True)| (since only one variable can
be changed at a time), but those are \lstinline|score 0|
regions which will never get accepted, forcing all our samples to be constant.}

\hide{To illustrate why formal irreducibility should perhaps not be an end goal in itself, notice that
we can make this particular kernel formally irreducible by replacing the \lstinline|score 0| 
with an extremely small non-zero value. But in practice our samples will still be
constant due to the tiny probability of accepting the intermediate sample.}

\hide{Another general approach would be to use Theorem~\ref{thm:state-dependent-mixing} to mix single-site MH with another kernel that is more reliably irreducible,
such as the kernel from Section~\ref{sec:msmh}.}

\hide{\subsection{Discussion}
\subsubsection{Differences between single-site MH and lightweight MH}
Although mostly similar, there are some key subtle differences---both theoretically and in
implementation---between lightweight MH and single-site MH.
Unlike in the original lightweight MH implementation, we delegate the task of obtaining
the current set of active sites to Haskell's lazy evaluator, rather than reimplementing
any ad-hoc evaluation methods ourselves.

Avoiding addressing issues, cite earlier section.

Any other notable LMH implementations to compare to?

\subsubsection{Suitability to various probabilistic models}
\label{sec:single-site:suitability}
 Like other inference methods in probabilistic programming, the effectiveness
of single-site MH depends on the kind of probabilistic model it is used on. 
As we have already noted, the irreducibility of the Markov chain generated
by \lstinline|mutateSite| will depend on the model under consideration and the way in which
it has been defined. It is usually possible to guarantee the irreducibility of the Markov
chain simply by redefining the program appropriately, but there is no mechanical procedure
to achieve this.

Probabilistically clustering points is a nice example where the usefulness of single-site
MH is visible. Assume we have a set of points (say, on the real line) and that we want
to assign to each point one of two Gaussian distributions on the real line (with some
fixed mean and variance).
For such a probabilistic program, the sites in the random tree will end up indicating
which cluster the corresponding points will belong to.
In terms of our proposal kernels this procedure amounts to proposing clusters for each
point by mutating tree sites.
We argue that in this case modifying only one site of the tree
at a time (i.e. changing cluster assignments one at a time) is more effective due to
its gradual nature of proposals, enabling us to arrive at more likely cluster assignments
more quickly.
Mutating multiple sites at once will still eventually converge to the same final
distribution, but useful progress made towards approaching a likely clustering could
easily be reset if too many points' clusters get reassigned at once.
Furthermore, care will need to be taken in choosing the best bias $p$ in MHG, which, at the
very least, will depend on the number of points present in the dataset. Such fine-tuning
of parameters is not necessary in single-site MH. The situation for MHG gets even worse
when the number of points can vary over time.

Single-site MH has its drawbacks too. It could, for example, be avoided in models
exhibiting high correlations between random variables because proposals differing only
by a single changed site could be more likely to be rejected (as opposed to when multiple
sites are changed).

\subsubsection{Modifications to single-site MHG}
To alleviate the above issue of correlated variables, one could choose modify the proposal
kernel to mutate more than a single site, but to still pick sites from a set of
evaluated sites. To our knowledge such kernels have not been studied before,
and so it remains to be seen how useful they could be from a theoretical point of view.
Fortunately, implementing such proposal kernels is not too difficult in our 
memory-inspecting framework, making LazyPPL a useful environment for testing out
such prototypes.

Going further one could even consider combining MHG with single-site MH, where before
each proposal step one could flip a coin to decide whether to independently
probabilistically mutate each site as in MHG or mutate just a single site. Clustering
problems and models with correlated variables can benefit from such a setting.
}


%% file: related.tex
\label{sec:relatedwork}
Our aim in this work is to study the power of types and laziness as a practical synthetic measure theory.
Our work is inspired by many other developments on the practical front.
\subsection{Laziness in probabilistic programming languages}
The Church project~\cite{goodman:church} is a major inspiration for our work. Although Church is an eager language, it has a primitive memoization construct (c.f.~\S\ref{sec:easyexamples:stochmem}). This leads to a programming style for lazy behaviour: instead of writing
\lstinline[mathescape]|do {x <- t; y <- u; z <- v; $\dots$}|
and expecting lazy evaluation, one can write (roughly)
\lstinline[mathescape]|f <- memoize(\i -> case i of {1 -> t; 2 -> u; 3 -> v}); $\dots$|
with eager evaluation, and use \lstinline|f 1|, \lstinline|f 2|, \lstinline|f 3|
in place of \lstinline|x|, \lstinline|y|, \lstinline|z| respectively.
Although this is an unusual programming style, it is usable nonetheless. 
Since Church is untyped, the precise connection with our metalanguage (\S\ref{sec:interface}) is unclear, and the semantics seems slightly different.
But the connection with
non-parametric statistics is heavily emphasized, for example in the analysis of stick-breaking~\cite{goodman:church,Roy-ChurchNP} and exchangeable primitives~\cite{wu-church}.
In summary, from a bird's eye view, our metalanguage (\S\ref{sec:interface}) and implementation~\cite{nonanongit} form a variation of Church with more idiomatic laziness, a type system and a semantics.

Languages such as Anglican~\cite{anglican}, WebPPL~\cite{dippl}, BayesDB~\cite{saad-bayesdb} and Turing~\cite{bloem-reddy-lazy} follow within the tradition of Church, exploring ideas from non-parametric statistics further.
The Birch language~\cite{birch} is class-based, transpiling to C++ 
(and so the connection to Section~\ref{sec:interface} is unclear),
but Birch heavily uses laziness and advanced control flow manipulations in its inference methods~\cite{delayed_sampling,lazy-object-copy}.

Beyond these examples, laziness has been explored in various aspects of probabilistic programming, dating back at least as far as the pioneering work by Koller et al.~\cite{pfeffer-lazy}, and more recently in the work on lazy factored inference in Figaro~\cite{pfeffer-lazyfunction}, and efficient implementation in delimited continuations through Hansei~(\cite{kiselyov-shan-ocaml}, which focuses on discrete distributions). 

\subsection{Other probabilistic programming work using Haskell and quasi-Borel spaces}
Various libraries have exploited Haskell for probabilistic programming.
Hakaru~\cite{narayanan2016probabilistic,Narayanan-shan2020,shan} provides a DSL with impressive symbolic inference methods. Stochaskell~\cite{roberts} provides a DSL \hide{, embedded in Haskell,}which compiles to Stan, Church and other back-ends. Stochaskell moreover allows a limited form of lazy lists, implemented via Church's memoization.

Our work here is most heavily inspired by MonadBayes~\cite{monadbayes}, which is a monad-based implementation of a variety of inference combinators, also inspired by the formalism of quasi-Borel spaces~\cite{scibior}. Originally, MonadBayes was not fully lazy: the Metropolis-Hastings simulation was based on the state monad and did not support laziness. The LazyPPL project grew out of adding laziness to MonadBayes, leading to the developments and the more natural expression of the various examples in this paper. For simplicity, we have focused here on the Metropolis-Hastings inference, but in practice it would be appropriate to adapt other MonadBayes inference combinators to the lazy setting. This was explored in a recent version of MonadBayes, using a portion of the LazyPPL implementation \cite{monadbayesWebsite,tweag-blog-post}.

Going beyond the inference combinators of MonadBayes, quasi-Borel spaces are also a foundation for a new dependent type system based on `trace types'~\cite{lew-trace-types} (prototyped in Haskell). This provides a well-typed account of the `programmable inference' that makes recent languages such as Gen~\cite{gen} and Pyro~\cite{bingham2018pyro} so powerful in practice. Wasabaye~\cite{wasabaye} connects trace types to Haskell's type-level strings. A possibly fruitful direction would be to generalize the traces allowed in trace types to accommodate the laziness and rose-tree-based sample space that we use in this paper. 

Further beyond our aims here, quasi-Borel spaces have also found profit in many other areas of probabilistic programming, including program logics (e.g.~\cite{qbs-logic,prob-log-ho-adversarial}) and functional languages for probabilistic network verification (\cite{Plonk}). 

\subsection{Other implementations of synthetic probability theory}
Finally we note two other approaches to metalanguages for categorical probability theory.
The first is the EfProb library~\cite{efprob}, a python library inspired by the effectus theory foundation for probability~\cite{effectus}.
The second is a python/F\# library for exact conditioning over Gaussian-based models~\cite{gaussianinfer}, inspired by categorical constructions over
Markov categories~\cite{dario,fritz}.
These approaches are currently focused on more refined notions of conditional probability, 
in contrast to our approach which is based on the measure-theoretic foundations of general purpose Monte Carlo-based inference.


%% file: paper.bbl
\begin{thebibliography}{10}

\bibitem{monadbayesWebsite}
Monad-{{Bayes}}.
\newblock \url{https://monad-bayes-site.netlify.app/_site/about.html}.

\bibitem{aafrr-comput-cond-ind}
N.~L. Ackerman, J.~Avigad, C.~E. Freer, D.~M. Roy, and J.~M. Rute.
\newblock Algorithmic barriers to representing conditional independence.
\newblock In {\em Proc.~LICS 2019}.

\bibitem{xrp-roy}
N.~L. Ackerman, C.~E. Freer, and D.~M. Roy.
\newblock Exchangeable random primitives.
\newblock In {\em Proceedings of the Workshop on Probabilistic Programming
  Semantics}, 2016.

\bibitem{prob-log-ho-adversarial}
A.~Aguirre, G.~Barthe, D.~Garg, M.~Gaboardi, S.~Katsumata, and T.~Sato.
\newblock Higher-order probabilistic adversarial computations: categorical
  semantics and program logics.
\newblock In {\em Proc.~ICFP 2021}, 2021.

\bibitem{aumann}
R.~J. Aumann.
\newblock Borel structures for function spaces.
\newblock {\em Illinois Journal of Mathematics}, 5, 1961.

\bibitem{lebesgue-R}
R.~Baker.
\newblock Lebesgue measure on $\mathbf{R}^\infty$.
\newblock {\em Proc.~AMS}, 113(4), 1991.

\bibitem{bingham2018pyro}
E.~Bingham, J.~P. Chen, M.~Jankowiak, F.~Obermeyer, N.~Pradhan, T.~Karaletsos,
  R.~Singh, P.~Szerlip, P.~Horsfall, and N.~D. Goodman.
\newblock {Pyro: Deep Universal Probabilistic Programming}.
\newblock {\em Journal of Machine Learning Research}, 2018.

\bibitem{bloem-reddy-lazy}
B.~Bloem-Reddy, E.~Mathieu, A.~Foster, T.~Rainforth, Y.~W. Teh, M.~Lomeli,
  H.~Ge, and Z.~Ghahramani.
\newblock Sampling and inference for discrete random probability measures in
  probabilistic programs.
\newblock In {\em Proc.~NeurIPS 2017 Workshop on Advances in Approximate
  Bayesian Inference}, 2017.

\bibitem{bdgs-lambdacalc-prob}
J.~Borgstrom, U.~{Dal~Lago}, A.~D. Gordon, and M.~Szymczak.
\newblock A lambda-calculus foundation for universal probabilistic programming.
\newblock In {\em Proc.~ICFP 2016}, 2016.

\bibitem{stan}
B.~Carpenter, A.~Gelman, M.~D. Hoffman, D.~Lee, B.~Goodrich, M.~Betancourt,
  M.~Brubaker, J.~Guo, P.~Li, and A.~Riddell.
\newblock Stan: A probabilistic programming language.
\newblock {\em Journal of statistical software}, 76(1), 2017.

\bibitem{efprob}
K.~Cho and B.~Jacobs.
\newblock The {E}f{P}rob library for probabilistic calculations.
\newblock In {\em Proc.~CALCO 2017}, 2017.

\bibitem{cho-jacobs}
K.~Cho and B.~Jacobs.
\newblock Disintegration and {B}ayesian inversion via string diagrams.
\newblock {\em Math. Struct. Comput. Sci.}, 29:938--971, 2019.

\bibitem{effectus}
K.~Cho, B.~Jacobs, B.~Westerbaan, and A.~Westerbaan.
\newblock An introduction to effectus theory.
\newblock arxiv:1512.05813, 2015.

\bibitem{quickcheck}
K.~Claessen and J.~Hughes.
\newblock Quick{C}heck: a lightweight tool for random testing of {H}askell
  programs.
\newblock In {\em Proc.~ICFP 2000}, pages 268--279, 2000.

\bibitem{coecke-terminality}
B.~Coecke.
\newblock Terminality implies non-signalling.
\newblock In {\em Proc.~QPL 2014}, 2014.

\bibitem{gen}
M.~F. Cusumano-Towner, F.~A. Saad, A.~K. Lew, and V.~K. Mansinghka.
\newblock Gen: a general-purpose probabilistic programming system with
  programmable inference.
\newblock pages 221--236, 2019.

\bibitem{deep-gp}
A.~Damianou and N.~D. Lawrence.
\newblock Deep {G}aussian processes.
\newblock In {\em Proc.~AISTATS 2013}, 2013.

\bibitem{nonanongit}
S.~Dash, Y.~Kaddar, H.~Paquet, and S.~Staton.
\newblock Lazy{P}{P}{L}: Lazy {P}robabilistic {P}rogramming {L}ibrary.
\newblock https://lazyppl.bitbucket.io/, 2022.
\newblock Code repository and web page with examples.

\bibitem{swaraj}
S.~Dash and S.~Staton.
\newblock A monad for probabilistic point processes.
\newblock In D.~I. Spivak and J.~Vicary, editors, {\em Proceedings of the 3rd
  Annual International Applied Category Theory Conference 2020, {ACT} 2020,
  Cambridge, USA, 6-10th July 2020}, volume 333 of {\em {EPTCS}}, pages 19--32,
  2020.

\bibitem{measurable-cones}
T.~Ehrhard, M.~Pagani, and C.~Tasson.
\newblock Measurable cones and stable, measurable functions: a model for
  probabilistic higher-order programming.
\newblock In {\em Proc.~POPL 2018}, 2018.

\bibitem{fritz}
T.~Fritz.
\newblock A synthetic approach to {M}arkov kernels, conditional independence
  and theorems on sufficient statistics.
\newblock {\em Adv.~Math.}, 370, 2020.

\bibitem{fgp-definetti}
T.~Fritz, T.~Gonda, and P.~Perrone.
\newblock De {F}inetti's theorem in categorical probability.
\newblock {\em Journal of Stochastic Analysis}, 2(4), 2021.

\bibitem{representable_markov}
T.~Fritz, T.~Gonda, P.~Perrone, and E.~F. Rischel.
\newblock Representable {M}arkov categories and comparison of statistical
  experiments in categorical probability, 2020.

\bibitem{fl-free-gs}
T.~Fritz and W.~Liang.
\newblock Free gs-monoidal categories and free {M}arkov categories.
\newblock arXiv:2204.02284, April 2022.

\bibitem{infinite_products}
T.~Fritz and E.~F. Rischel.
\newblock Infinite products and zero-one laws in categorical probability.
\newblock {\em {Compositionality}}, 2, Aug. 2020.

\bibitem{geyer}
C.~Geyer.
\newblock Introduction to {M}arkov {C}hain {M}onte {C}arlo.
\newblock In {\em Handbook of Markov Chain Monte Carlo}. Chapman Hall/CRC,
  2011.

\bibitem{gv-bnp}
S.~Ghosal and A.~van~der Vaart.
\newblock {\em Fundamentals of non-parametric {B}ayesian inference}.
\newblock CUP, 2017.

\bibitem{giry}
M.~Giry.
\newblock A categorical approach to probability theory.
\newblock {\em Categorical Aspects of Topology and Analysis. Lecture Notes in
  Mathematics}, 1982.

\bibitem{goodman:church}
N.~Goodman, V.~Mansinghka, D.~M. Roy, K.~Bonawitz, and J.~B. Tenenbaum.
\newblock Church: a language for generative models.
\newblock 2008.

\bibitem{dippl}
N.~D. Goodman and A.~Stuhlm\"{u}ller.
\newblock {The Design and Implementation of Probabilistic Programming
  Languages}.
\newblock \url{http://dippl.org}, 2014.
\newblock Accessed: 2020-10-15.

\bibitem{jgl-domain-probprog}
J.~Goubault-Larrecq, X.~Jia, and C.~Théron.
\newblock A domain-theoretic approach to statistical programming languages.
\newblock arxiv:2106.16190, June 2021.

\bibitem{rjmcmc}
P.~J. Green.
\newblock Reversible jump {M}arkov {C}hain {M}onte {C}arlo computation and
  {B}ayesian model determination.
\newblock {\em Biometrika}, 82(4):711--732, 1995.

\bibitem{ibp}
T.~Griffiths and Z.~Ghahramani.
\newblock The {I}ndian buffet process: An introduction and review.
\newblock {\em Journal of Machine Learning Research}, 12(32):1185--1224, 2011.

\bibitem{qbs}
C.~Heunen, O.~Kammar, S.~Staton, and H.~Yang.
\newblock A convenient category for higher-order probability theory.
\newblock In {\em Proc.~LICS 2017}, 2017.

\bibitem{Hinze00memofunctions}
R.~Hinze.
\newblock Memo functions, polytypically!
\newblock In {\em Proceedings of the 2nd Workshop on Generic Programming, Ponte
  de}, pages 17--32, 2000.

\bibitem{hms-computable-ppl}
D.~Huang, G.~Morrisett, and B.~Spitters.
\newblock An application of computable distributions to the semantics of
  probabilistic programs.
\newblock arxiv:1806.07966, 2018.

\bibitem{jacobs-weakening}
B.~Jacobs.
\newblock Semantics of weakening and contraction.
\newblock {\em Ann. Pure Appl. Logic}, 69, 1994.

\bibitem{jacobs-probabilities}
B.~Jacobs.
\newblock Probabilities, distribution monads, and convex categories.
\newblock {\em Theoret.~Comput.~Sci.}, 412, 2011.

\bibitem{kechris}
A.~Kechris.
\newblock {\em Classical Descriptive Set Theory}.
\newblock Springer, 1987.

\bibitem{kiselyov-shan-ocaml}
O.~Kiselyov and C.~Shan.
\newblock Embedded probabilistic programming.
\newblock In {\em Proc.~DSL 2009}, 2009.

\bibitem{kock-comm-monad}
A.~Kock.
\newblock Monads on symmetric monoidal closed categories.
\newblock {\em Arch.~Math.}, 21, 1970.

\bibitem{kock-affine}
A.~Kock.
\newblock Bilinearity and cartesian closed monads.
\newblock {\em Math.~Scand.}, 29, 1971.

\bibitem{kock}
A.~Kock.
\newblock Commutative monads as a theory of distributions.
\newblock {\em Theory and Applications of Categories}, 26(4), 2012.

\bibitem{pfeffer-lazy}
D.~Koller, D.~McAllester, and A.~Pfeffer.
\newblock Effective {B}ayesian inference for stochastic programs.
\newblock In {\em Proc.~AAAI 1997}, 1997.

\bibitem{kozen}
D.~Kozen.
\newblock Semantics of probabilistic programs.
\newblock {\em Journal of Computer and System Sciences}, 22:328--350, 1981.

\bibitem{geometryDLH}
U.~D. Lago and N.~Hoshino.
\newblock The geometry of bayesian programming.
\newblock In {\em 34th Annual {ACM/IEEE} Symposium on Logic in Computer
  Science, {LICS} 2019, Vancouver, BC, Canada, June 24-27, 2019}, pages 1--13.
  {IEEE}, 2019.

\bibitem{freyd-cat}
P.~B. Levy, J.~Power, and H.~Thielecke.
\newblock Modelling environments in call-by-value programming languages.
\newblock {\em Inform.~Comput.}, 185:182--210, 2003.

\bibitem{lew-trace-types}
A.~K. Lew, M.~F. Cusumano-Towner, B.~Sherman, M.~Carbin, and V.~K. Mansinghka.
\newblock Trace types and denotational semantics for sound programmable
  inference in probabilistic languages.
\newblock In {\em Proc.~POPL 2020}, 2020.

\bibitem{bugs}
D.~Lunn, D.~Spiegelhalter, A.~Thomas, and N.~Best.
\newblock The {B}{U}{G}{S} project: Evolution, critique and future directions.
\newblock {\em Statistics in Medicine}, 28(25):3049--3067, 2009.

\bibitem{maraist1999call}
J.~Maraist, M.~Odersky, D.~N. Turner, and P.~Wadler.
\newblock Call-by-name, call-by-value, call-by-need and the linear lambda
  calculus.
\newblock {\em Theoretical Computer Science}, 228(1-2):175--210, 1999.

\bibitem{memo-nature}
D.~Michie.
\newblock '{M}emo' functions and machine learning.
\newblock {\em Nature}, 218, 1968.

\bibitem{moggi:computation_and_monads}
E.~Moggi.
\newblock Notions of computation and monads.
\newblock {\em Information and Computation}, 1991.

\bibitem{delayed_sampling}
L.~Murray, D.~Lund\'{e}n, J.~Kudlicka, D.~Broman, and T.~Sch\"{o}n.
\newblock Delayed sampling and automatic {R}ao-{B}lackwellization of
  probabilistic programs.
\newblock In {\em Proceedings of the Twenty-First International Conference on
  Artificial Intelligence and Statistics}, pages 1037--1046, 2018.

\bibitem{lazy-object-copy}
L.~M. Murray.
\newblock Lazy object copy as a platform for population-based probabilistic
  programming.
\newblock arxiv:2001.05293, Jan 2020.

\bibitem{birch}
L.~M. Murray and B.~Sch{\"o}n.
\newblock Automated learning with a probabilistic programming language: Birch.
\newblock {\em Annual Reviews in Control}, 2018.

\bibitem{narayanan2016probabilistic}
P.~Narayanan, J.~Carette, W.~Romano, C.~Shan, and R.~Zinkov.
\newblock Probabilistic inference by program transformation in {H}akaru (system
  description).
\newblock In {\em Proc.~FLOPS 2016}, pages 62--79, 2016.

\bibitem{Narayanan-shan2020}
P.~Narayanan and C.~Shan.
\newblock Symbolic disintegration with a variety of base measures.
\newblock {\em ACM Transactions on Programming Languages and Systems}, 42(2),
  2020.

\bibitem{navarroNonparametricBayesianMethod2006}
D.~Navarro and T.~Griffiths.
\newblock A nonparametric {B}ayesian method for inferring features from
  similarity judgments.
\newblock In B.~Sch\"olkopf, J.~Platt, and T.~Hoffman, editors, {\em Advances
  in Neural Information Processing Systems}, volume~19. {MIT Press}.

\bibitem{wasabaye}
M.~Nguyen, R.~Perera, M.~Wang, and N.~Wu.
\newblock Modular probabilistic models via algebraic effects.
\newblock In {\em Proc.~ICFP 2022}, 2022.
\newblock To appear.

\bibitem{BetaStickBreaking}
J.~W. Paisley, D.~M. Blei, and M.~I. Jordan.
\newblock Stick-breaking beta processes and the poisson process.
\newblock In N.~D. Lawrence and M.~A. Girolami, editors, {\em Proceedings of
  the Fifteenth International Conference on Artificial Intelligence and
  Statistics, {AISTATS} 2012, La Palma, Canary Islands, Spain, April 21-23,
  2012}, volume~22 of {\em {JMLR} Proceedings}, pages 850--858. JMLR.org, 2012.

\bibitem{pfeffer-lazyfunction}
A.~Pfeffer, B.~Ruttenberg, A.~Sliva, M.~Howard, and G.~Takata.
\newblock Lazy factored inference for functional probabilistic programming.
\newblock arxiv:1509.03564.

\bibitem{tweag-blog-post}
{Reuben Cohn-Gordon}.
\newblock Improving the probabilistic programming language {Monad-Bayes}.
\newblock \url{https://www.tweag.io/blog/2022-10-18-monad-bayes-fellowship/}.

\bibitem{roberts}
D.~A. Roberts, M.~Gallagher, and T.~Taimre.
\newblock Reversible jump probabilistic programming.
\newblock In K.~Chaudhuri and M.~Sugiyama, editors, {\em Proceedings of the
  Twenty-Second International Conference on Artificial Intelligence and
  Statistics}, volume~89 of {\em Proceedings of Machine Learning Research},
  pages 634--643. PMLR, 16--18 Apr 2019.

\bibitem{Roy-ChurchNP}
D.~Roy, V.~Mansinghka, N.~Goodman, and J.~Tenenbaum.
\newblock A stochastic programming perspective on nonparametric {B}ayes.
\newblock In {\em Proc.~Workshop on Non-Parametric Bayes}, 2008.

\bibitem{roy2014continuum}
D.~M. Roy.
\newblock The continuum-of-urns scheme, generalized beta and indian buffet
  processes, and hierarchies thereof.
\newblock {\em arXiv preprint arXiv:1501.00208}, 2014.

\bibitem{roy-ibp}
D.~M. Roy, N.~Ackerman, J.~Avigad, C.~Freer, and J.~Rute.
\newblock Exchangeable graphs, conditional independence, and
  computably-measurable samplers.
\newblock Talk at CCA 2013.
\newblock \url{http://cca-net.de/cca2013/slides/17_Daniel%20Roy.pdf}.

\bibitem{saad-bayesdb}
F.~Saad and V.~Mansinghka.
\newblock Detecting dependencies in sparse, multivariate databases using
  probabilistic programming and non-parametric {B}ayes.
\newblock In {\em Proc.~AISTATS 2017}, 2017.

\bibitem{qbs-logic}
T.~Sato, A.~Aguirre, G.~Barthe, D.~Garg, M.~Gaboardi, and J.~Hsu.
\newblock Formal verification of higher-order probabilistic programs.
\newblock In {\em Proc.~POPL 2019}, 2019.

\bibitem{monadbayes}
A.~\'Scibior, O.~Kammar, and Z.~Ghahramani.
\newblock Functional programming for modular {B}ayesian inference.
\newblock In {\em Proc.~ICFP 2018}, 2018.

\bibitem{scibior}
A.~{\'S}cibior, O.~Kammar, M.~Vákár, S.~Staton, H.~Yang, Y.~Cai,
  K.~Ostermann, S.~K. Moss, C.~Heunen, and Z.~Ghahramani.
\newblock Denotational validation of higher-order {B}ayesian inference.
\newblock In {\em Proc.~POPL 2018}, 2018.

\bibitem{shiebler}
D.~Shiebler.
\newblock Categorical stochastic processes and likelihood.
\newblock In {\em Proc.~ACT 2020}, 2020.

\bibitem{staton:sfinite}
S.~Staton.
\newblock Commutative semantics for probabilistic programming.
\newblock In {\em Proc.~ESOP 2017}, 2017.

\bibitem{staton-probprog-measure}
S.~Staton.
\newblock Probabilistic programs as measures.
\newblock In {\em Foundations of Probabilistic Programming}. CUP, 2020.

\bibitem{splitmix}
G.~L. {Steele~Jr}, D.~Lea, and C.~H. Flood.
\newblock Fast splittable pseudorandom number generators.
\newblock In {\em Proc.~OOPSLA 2014}, 2014.

\bibitem{gaussianinfer}
D.~Stein.
\newblock Gaussian{I}nfer.
\newblock \url{https://github.com/damast93/GaussianInfer}, 2021.

\bibitem{dario-thesis}
D.~Stein.
\newblock {\em Structural Foundations for Probabilistic Programming Languages}.
\newblock PhD thesis, University of Oxford, 2021.

\bibitem{dario}
D.~Stein and S.~Staton.
\newblock Compositional semantics for probabilistic programs with exact
  conditioning.
\newblock In {\em Proc.~LICS 2021}, 2021.

\bibitem{thibaux-jordan}
R.~Thibaux and M.~I. Jordan.
\newblock Hierarchical beta processes and the {I}ndian buffet process.
\newblock In {\em Proc.~AISTATS 2007}, 2007.

\bibitem{tierney-mcposterior}
L.~Tierney.
\newblock Markov chains for exploring posterior distributions.
\newblock {\em The Annals of Statistics}, 22(4):1701--1728, 1994.

\bibitem{anglican}
D.~Tolpin, H.~Yang, J.~W. van~de Meent, and F.~Wood.
\newblock Design and implementation of probabilistic programming language
  {A}nglican.
\newblock 2016.

\bibitem{meent-book}
J.-W. van~de Meent, B.~Paige, H.~Yang, and F.~Wood.
\newblock An introduction to probabilistic programming.
\newblock 2018.

\bibitem{Plonk}
A.~Vandenbroucke and T.~Schrijvers.
\newblock Plonk: functional probabilistic {N}et{K}at.
\newblock In {\em Proc.~POPL 2020}, 2020.

\bibitem{wqbs}
M.~Vákár, O.~Kammar, and S.~Staton.
\newblock A domain theory for statistical probabilistic programming.
\newblock In {\em Proc.~POPL 2019}, 2019.

\bibitem{shan}
R.~Walia, P.~Narayanan, J.~Carette, S.~Tobin-Hochstadt, and C.~Shan.
\newblock From high-level inference algorithms to efficient code.
\newblock In {\em Proc.~ICFP 2019}, 2019.

\bibitem{lightweight-mh}
D.~Wingate, A.~Stuhlmueller, and N.~Goodman.
\newblock Lightweight implementations of probabilistic programming languages
  via transformational compilation.
\newblock In {\em Proc.~AISTATS 2011}, 2011.

\bibitem{wood-sequence-memoizer}
F.~D. Wood, C.~Archambeau, J.~Gasthaus, L.~James, and Y.~W. Teh.
\newblock A stochastic memoizer for sequence data.
\newblock In {\em Proc.~ICML 2009}, 2009.

\bibitem{wu-church}
J.~Wu.
\newblock Reduced traces and {J}{I}{T}ing in church.
\newblock Master's thesis, MIT, 2013.

\end{thebibliography}
